\ifdefined\LaTeXdiff
	\let\WithReply\undefined
\fi

\RequirePackage{fix-cm}

\let\saveopenmath\(
\documentclass{lmcs} 
\let\(\saveopenmath

\pdfoutput=1

\usepackage{lastpage}
\lmcsdoi{16}{1}{5}
\lmcsheading{}{\pageref{LastPage}}{}{}%
{Jul.~20,~2018}{Jan.~22,~2020}{}

\keywords{timed systems, timing parameters, language preservation, trace preservation, L/U-PTA}

\usepackage{nico}

\usepackage[nolineno,nocomments]{inverse}

\usepackage[utf8]{inputenc}
\usepackage[english]{babel}
\usepackage{wrapfig}

\usepackage[linesnumbered,ruled,vlined]{algorithm2e}
	\SetKwInOut{Input}{input}
	\SetKwInOut{Output}{output}

\usepackage{tikz}
\usetikzlibrary{automata, arrows, backgrounds, positioning}
\tikzstyle{location}=[minimum size=12pt, circle, fill=blue!20, 
  draw=black, text=black, inner sep=1.5pt, initial text={}] 
\tikzstyle{invariant}=[rectangle, draw=black, text=black, yshift=10]

\makeatletter
\def\cl@chapter{}
\makeatother

\usepackage[capitalise,english,nameinlink]{cleveref}
\crefname{line}{\text{line}}{\text{lines}} 

\makeatletter
\@ifundefined{lemma}{\theoremstyle{plain}%
  \newtheorem{theorem}{Theorem}[section]
  \newtheorem{lemma}[theorem]{Lemma}
  \newtheorem{definition}[theorem]{Definition}
  \newtheorem{remark}[theorem]{Remark}
  \newtheorem{example}[theorem]{Example}
  
  \newtheorem{proposition}[theorem]{Proposition}}{}
\let\c@definition\c@theorem
\let\c@lemma\c@theorem
\let\c@corollary\c@theorem
\let\c@remark\c@theorem
\let\c@example\c@theorem
\let\c@proposition\c@theorem
\makeatother

\begin{document}

\title[Language Preservation Problems in Parametric Timed Automata]{Language Preservation Problems in Parametric Timed Automata}
\titlecomment{{\lsuper*}This work is partially supported by the ANR national research program
        PACS (ANR-14-CE28-0002),
by European projects ERC EQualIS (308087) and FET Cassting (601148),
by the ANR-NRF French-Singaporean research program ProMiS (ANR-19-CE25-0015)
	and by
	ERATO HASUO Metamathematics for Systems Design Project (No.\ JPMJER1603), JST.
This work is an extension of~\cite{AM15}.
}

\author[\'E.~Andr\'e]{\'Etienne~Andr\'e}	
\address{Université Paris 13, LIPN, CNRS, UMR 7030, F-93430, Villetaneuse, France}
\address{JFLI, CNRS, Tokyo, Japan}
\address{National Institute of Informatics, Tokyo, Japan}
\address{Université de Lorraine, CNRS, Inria, LORIA, Nancy, France}
\urladdr{\url{https://lipn.univ-paris13.fr/~andre/}}  

\author[D.~Lime]{Didier~Lime}	
\address{École Centrale de Nantes, LS2N, CNRS, UMR 6004, Nantes, France}	
\urladdr{\url{http://pagesperso.ls2n.fr/~lime-d/}}  

\author[N.~Markey]{Nicolas~Markey}	
\address{IRISA, CNRS \& Inria \& Univ. Rennes, France}	
\urladdr{\url{http://people.irisa.fr/Nicolas.Markey/}}  


\maketitle

\ifcomments
	\textcolor{red}{\textbf{This is the version with comments. To disable comments, modify line~3 in the \LaTeX{} source.}}
\fi


\ifcomments
\fi


\begin{abstract}
	Parametric timed automata (PTA) are a powerful formalism to model and
        reason about concurrent systems with some unknown timing delays. In
        this paper, we address the (untimed) language- and trace-preservation
        problems: \emph{given a reference parameter valuation, does there exist another parameter valuation with the same untimed language, or with the same set of traces?}
        We~show that these problems are undecidable both for
        general PTA and for the restricted class of L/U-PTA, even for integer-valued parameters, or over bounded time.
        On~the other hand, we~exhibit decidable subclasses: 1-clock PTA, and 1-parameter deterministic L-PTA and U-PTA. 
        We~also consider robust versions of these problems, where
		we~additionally require that the language be preserved for all valuations between the reference valuation and the new valuation.
\end{abstract}


\maketitle

\section{Introduction}

\paragraph{Timed Automata.}
Timed Automata (TA hereafter) were introduced in the 1990's
\cite{AD90} as an extension of finite automata with \emph{clock
  variables}, which can be used to constrain the delays between
transitions. Despite this flexibility, TA~enjoy efficient algorithms
for checking reachability (and many other properties), which makes
them a well-suited model for reasoning about real-time systems.

In~TA, clock variables are compared to (integer) constants in order to
allow or disallow certain transitions. The behaviour of a TA may
heavily depend on the \emph{exact} values of the constants, and slight
changes in any constant may give rise to very different
behaviours. In~many cases however, it~may be desirable to optimise the
values of some of the constants of the automaton, in order to exhibit
better performances. The question can then be posed as follows:
\emph{given a TA and some of its integer constants, does there exist
  other values of these constants for which the TA has the exact set
  of (untimed) behaviours?} We~call this problem the
\emph{language-preservation problem}.

A special case of this problem occurs naturally in recent approaches
for dealing with \emph{robustness} of timed automata
\cite{DDMR08,Sankur-MFCS11,San13}. The question asked there is whether
the behaviour of a timed automaton is preserved when the clock
constraints are slightly (parametrically) enlarged.  In~most of those
cases, the existence of a parametric enlargement for which the
behaviours are the same as in the original TA has been proved
decidable.

For the general problem however, the decidability status remains open. To~the
best of our knowledge, the only approach to this problem is a procedure
(called the \emph{inverse method} \cite{ACEF09}) to compute a dense set of
parameter valuations around a reference valuation~$\pval_0$.

\paragraph{Parametric Timed Automata.}
In this paper, we address the language-\hskip0pt preservation problem using
\emph{Parametric Timed Automata}~(PTA) \cite{AHV93}. A~PTA is a TA in which
some of the numerical constants in clock constraints are replaced by symbolic constants (a.k.a.\ parameters), whose value is not known a~priori. The classical problem
(sometimes called the \emph{EF-emptiness problem}) in PTA asks whether a given target
location of a PTA is reachable for some valuation of the parameter(s). This
problem was proven undecidable in various settings: for integer parameter
valuations \cite{AHV93,BBLS15}, for bounded rational valuations \cite{Miller00}, or with only strict constraints (no equality nor closed inequality) \cite{Doyen07}.
The proofs of these results exist in many different flavours, with various bounds on the number of parameters and clocks needed in the reductions; in contrast, limiting the number of clocks (see \eg{} \cite{AHV93,BO14,BBLS15}) yields decidability (see \cite{Andre15} for a survey).

The only non-trivial syntactic subclass of PTA
with decidable 
EF-emptiness problem is the class of
L/U-PTA \cite{HRSV02}. These models have the following constraint: each
parameter may only be used either always as a lower bound in the clock
constraints, or always as an upper bound.
For those models, the problems of the emptiness, universality and finiteness (for integer-valued parameters) of the set of parameters under which a target location is
reachable, are decidable \cite{HRSV02,BlT09}. In~contrast, the AF-emptiness problem
(``\emph{does there exist a parameter valuation for which a given location is   eventually visited along any run?}'')\ is undecidable for L/U-PTA \cite{JLR14}.
The EG-emptiness problem (``\emph{does there exist a parameter valuation for which a maximal path remains permanently within a given set of locations?}'')\ exhibits a thin border between decidability and undecidability: the problem is decidable if and only if (rational-valued) parameters are chosen in a closed interval~\cite{ALime17}.
The full TCTL logic-emptiness (``\emph{does there exist a parameter valuation for which a given TCTL formula holds?}'')\ is undecidable for the simpler class of U-PTA~\cite{ALR18FORMATS}, where parameters can only be used as upper bounds in clock constraints.

\paragraph{Our Contributions.}
In~this paper, we first prove that the language-preservation problem
(and various related problems) is undecidable in most
cases (including for L/U-PTA, or in the time-bounded setting). While it might not look surprising given the numerous 
undecidability results about PTA, it~contrasts with the decidability
results proved so far for robustness of~TA. In~the parametrized
approaches to robustness (where the aim is to decide if the language
of a timed automaton is preserved under a parametrized perturbation)
\cite{DDMR08,Sankur-MFCS11,San13}, the use of the parameter is much
more constrained than what we allow in this paper; this is what makes
parametrized robustness analysis decidable.

We then devise a semi-algorithm that solves the language-
and trace-preservation 
problems (and actually synthesizes all parameter valuations yielding the same
untimed language (or~trace) as a given reference valuation), in the setting of
\emph{deterministic} PTA.
Finally, we~study the decidability of these problems
for subclasses of~PTA: we~prove decidability for PTA with a single clock,
and for two subclasses of L/U-PTA
with a single parameter.


\paragraph{Outline}
\cref{section:definitions} recalls the necessary preliminaries.
\cref{section:undecidable} proves the undecidability of the problems in
general. 
\cref{section:algo} introduces a correct semi-algorithm for the trace- and
language-preservation synthesis. 
\cref{section:particular} considers the (un)decidability for subclasses of PTA.
\cref{section:conclusion} concludes the paper.

\section{Definitions}\label{section:definitions}




\subsection{Constraints}

We fix a finite set~$\Clock = \{ \clock_1, \dots, \clock_\ClockCard \}$ set of
real-valued variables (called \emph{clocks} in the sequel). A~clock
valuation~$\clockval$ is a function $\clockval\colon \Clock \rightarrow \bbR+$. 
We denote by $\mathbf{0}_{\Clock}$ the clock valuation assigning $0$ to all clocks. We~define two operations on clock valuations: for $d\in\bbR+$ and a
clock valuation~$\clockval$, we~let $\clockval+d$ be the
valuation~$\clockval'$ such that $\clockval'(\clock) = \clockval(\clock)+d$
for all~$\clock \in\Clock$. Given a set~$R\subseteq \Clock$ and a
valuation~$\clockval$, we~let $\clockval[R\mapsto 0]$ be the clock
valuation~$\clockval'$ such that $\clockval'(\clock)=0$ if~$\clock\in R$, and
$\clockval'(\clock)=\clockval(\clock)$ otherwise.

We also fix a finite set~$\Param = \{ \param_1, \dots, \param_\ParamCard \} $
of rational-valued variables called \emph{parameters}.
A~parameter valuation~$\pval$ is a function $\pval\colon \Param
\rightarrow \bbQ+$.
In the sequel, we will have to handle clocks and parameters together.
A~valuation is a function $\val\colon \Clock\cup\Param\to\bbR+$ such that
$\val_{|\Clock}$ is a clock valuation and $\val_{|\Param}$ is a parameter
valuation.

An~\emph{atomic constraint} over $\Clock$ and~$\Param$ is an expression of the
form either $\clock \compOp \param + c$ or $\clock\compOp c$ or $\param \compOp c$,
where $\mathord{\compOp} \in \{\mathord<, \mathord\leq, \mathord=, \mathord\geq,
\mathord>\}$, $\clock\in\Clock$, $\param\in\Param$ and $c\in\bbZ$. The
symbols~$\KTrue$ and~$\KFalse$ are also special cases of atomic constraints.
Notice that our constraints are a bit more general than in the setting of~\cite{AHV93}, where
only atomic constraints of the form $\clock\compOp\param$ and $\clock\compOp c$
(and~$\KTrue$ and~$\KFalse$) were allowed. 
A~\emph{constraint} over $\Clock$ and~$\Param$ is a conjunction of atomic constraints.
An~\emph{(atomic) diagonal constraint} is a constraint of the form
$\clock-\clock'\compOp \param+c$ or $\clock-\clock'\compOp c$, where $\clock$ and
$\clock'$ are two clocks and $\compOp$, $\param$ and~$c$ are as in plain atomic
constraints. A~\emph{generalized constraint} over~$\Clock$ and~$\Param$ is a
conjunction of atomic constraints and atomic diagonal constraints.

\begin{remark}
  We mainly focus here on continuous time (where clock valuations take
  real values) and rational-valued parameters, as defined above.
  However, several of our results remain valid for discrete time (where
  clock valuations take integer values) and integer-valued parameters.
  We will mention it explicitly when such is the case.
\end{remark}

A valuation~$\val$ satisfies an atomic constraint~$\phi\colon \clock \compOp
\param+c$, which we denote $\val\models\phi$, whenever
$\val(\clock)\compOp\val(\param)+c$. The~definition for other constraints is similar.
All~valuations satisfy~$\KTrue$, and none of them satisfies~$\KFalse$.
A~valuation~$\val$ satisfies a
constraint~$\Phi$, denoted $\val\models\Phi$ if, and only~if, it satisfies all
the conjuncts of~$\Phi$.
%
A~constraint~$\Phi$ is said to depend on~$D\subseteq \Clock\cup\Param$ whenever
for any two valuations~$\val$ and~$\val'$ such that $\val(d)=\val'(d)$ for
all~$d\in D$, it~holds $\val\models\Phi$ if, and only~if, $\val'\models\Phi$.
A~parameter constraint is a constraint that depends only on~$\Param$.
%

Given a partial valuation~$\val$ and a constraint~$\Phi$,
we~write $\valuate{\Phi}{\val}$ for the constraint obtained by
replacing each~$z$ in the domain $\domain(\val)$ of $\val$  in~$\Phi$ with~$\val(z)$. The resulting
constraint depends on $(\Clock\cup\Param)\setminus \domain(\val)$.
%
%

\smallskip
We denote by $\project{\Phi}{V}$ the \emph{projection} of
constraint~$\Phi$ onto~$V\subseteq \Clock\cup \Param$, \ie{} the
constraint obtained by eliminating the variables not in~$V$.
Satisfaction of a projected constraint is defined as:
$\pval\models\project{\Phi}{V}$ if, and only~if, there exists a
valuation~$\val$ on $\Clock\cup\Param$
such that $\val\models\Phi$ and $\val_{|V} = \pval$.
In~particular, we~will be interested in the projection onto the
set~$\Param$ of parameters.  Such projections can be computed \eg{}
using Difference Bound Matrices~(DBM)~\cite{BY03},
or Fourier-Motzkin elimination.
%
We~also define the \emph{time elapsing} of~$\Phi$, denoted by
$\timelapse{\Phi}$, as the \emph{generalized} constraint over~$\Clock$
and~$\Param$ obtained from~$\Phi$ by delaying an arbitrary amount of
time: $\pval\models\timelapse{\Phi}$ if, and only~if, there exists a
valuation~$\val$ on $\Clock\cup\Param$ and a delay~$d\in\bbR+$ such
that $\val\models\Phi$ and $\pval_{|\Param}=\val_{|\Param}$ and
$\pval_{|\Clock}=\val_{|\Clock}+d$.
The~time-elapsing of a
constraint~$\Phi$ is a classical computation using polyhedra or parametric extensions of~DBM: it~can be
obtained by preserving all differences between any pair of clocks,
preserving lower bounds, relaxing upper bounds on atomic
(single-clock) constraints, and preserving all relations between
parameters (and~constants).
%
Given $\resets \subseteq \Clock$, we define the \emph{reset}
of~$\Phi$, denoted by $\reset{\Phi}{\resets}$, as the constraint over
$\Clock$ and $\Param$ obtained from~$\Phi$ by resetting all clocks
in~$\resets$. Its~satisfaction relation is defined as follows:
$\pval\models\reset{\Phi}{\resets}$ if, and only~if, there exists a
valuation~$\val$ on $\Clock\cup\Param$ such that $\val\models\Phi$ and
$\val_{|\resets}=\mathbf{0}_{\resets}$ and
$\val_{|(\Param\cup\Clock\setminus\resets)}=\pval_{|(\Param\cup\Clock\setminus\resets)}$.
This~is again easily computed using polyhedra, DBM or Fourier-Motzkin elimination.
%

\begin{example}\label{example:operations-on-constraints}
	Assume $\Clock = \{ \clock_1 , \clock_2 \}$ be a set of clocks and $\Param = \{ \param_1 , \param_2 \}$ be a set of parameters.
	Consider the constraint (involving diagonal constraints) $\Phi$ defined by
	\[\Phi\equiv
        (\clock_1 = \param_1)
			\land
		(\param_1 > \param_2)
			\land
		(\clock_2 = \clock_1 - 1)
			\land
		(\clock_2 = 3)
		\]
	(the fact that all variables are non-negative is left implicit here).
	        Then, we have
                \[
                \projectP{\Phi} \equiv (\param_1 > \param_2) \land (\param_1 = 4).
                \]
	In~addition, resetting clock~$x_2$ in~$\Phi$ gives:
	\[\reset{\Phi}{\{ \clock_2 \}} \equiv
        (\clock_1 = \param_1)
			\land
		(\param_1 > \param_2)
			\land
		(\clock_2 = 0)
			\land
		(\clock_1 = 4).
		\]
	        Finally, letting time elapse from valuations satisfying~$\Phi$ gives:
		\[\timelapse{\Phi}\equiv (\clock_1 \geq \param_1)
			\land
		(\param_1 > \param_2)
			\land
		(\clock_2 = \clock_1 - 1)
			\land
		(\clock_2 \geq 3)
			\land
		(\param_1 = 4).
		\]
\end{example}

\subsection{Syntax of Parametric Timed Automata}

Parametric timed automata are an extension of the class of timed
automata to the parametric case, where parameters can be used
within guards and invariants in place of constants~\cite{AHV93}.

\begin{definition}
  \label{def:PTA}
  A~\emph{parametric timed automaton} (PTA for short) is a tuple $\calA =
  \tuple{\Sigma, \Loc, \locinit, \Clock, \Param, \invariant, \steps}$, where:
$\Sigma$ is a finite set of actions;
$\Loc$ is a finite set of locations;
$\locinit \in \Loc$ is the initial location;
$\Clock$ is a finite set of clocks;
$\Param$ is a finite set of parameters;
$\invariant$ assigns to every $\loc\in \Loc$ a constraint $\invariant(\loc)$, called the
                  \emph{invariant} of~$\loc$;
$\steps$ is a finite set of edges
                  $(\loc,\guard,\action,\resets,\loc')$, 
                  where
                  $\loc,\loc'\in \Loc$ are the source and destination
                  locations, $\guard$ is a constraint (called
                  \emph{guard} of the transition), $\action\in\Sigma$, and
                  $\resets\subseteq \Clock$ is a set of clocks to be reset.
\end{definition}

A PTA is \emph{deterministic} if, for all $\loc \in \Loc$, for all $\action
\in \Sigma$,
there is at most one edge
$(\loc', \guard ,\action',\resets,\loc'') \in \steps$
with $\loc'=\loc$ and $\action'=\action$.
Note that this is a stronger assumption than the usual definition of
determinism for~TA, which only requires that, for a given action,
guards must be mutually exclusive.

A clock is said to be \emph{parametric} if it is compared with a
parameter in at least one guard or invariant.  Otherwise, it is
\emph{non-parametric}.\label{newtext:parametric-clock}

\begin{figure}[tb]
\centering

\newcommand{\coulact}[1]{#1}
\newcommand{\coulclock}[1]{#1}
\newcommand{\coulparam}[1]{#1}
\newcommand{\parami}[1]{\coulparam{p_{#1}}}
	\definecolor{cv1}{rgb}{1, 0, 0}
	\definecolor{cv2}{rgb}{0, 1, 0}
	\definecolor{cv3}{rgb}{0, 0, 1}

\begin{tikzpicture}[->, >=stealth', node distance=4cm, thin]
\tikzstyle{state}=[circle, minimum size=12pt, draw, inner sep=1.5pt]
\tikzstyle{every node}=[initial text=]


	\node[state, initial, fill=cv2!20,inner sep=0pt] (Q0) {$\genfrac{}{}{0pt}{0}{\loc_1}{\phantom{\scriptstyle y\leq p_2}}$};
	\node[state, fill=cv3!20, right of=Q0,inner sep=0pt] (Q1) {$\genfrac{}{}{0pt}{0}{\loc_2}{\coulclock{y} \leq \parami{2}}$};
	\node[state, fill=cv1!20, right of=Q1,inner sep=0pt] (Q2) {$\genfrac{}{}{0pt}{0}{\loc_3}{\coulclock{y} \leq \parami{3}}$};



	\path (Q0) edge [below] node{\begin{tabular}{c}
		\coulact{press?} \\
		$\coulclock{x} := 0$ \\
		$\coulclock{y} := 0$ \\
		\end{tabular}} (Q1);
	\path (Q1) edge [] node[xshift=5] {\begin{tabular}{c}
		$\coulclock{y}=\parami{2}$ \\[2mm]
		\coulact{cup!} \\
		\end{tabular}} (Q2);
	\path [loop, out=-60,in=-120,looseness=4] (Q1) edge node[below]
                {\begin{tabular}{c}
		$\coulclock{x}\geq\parami{1}$ \\
		\coulact{press?} \\
		$\coulclock{x}:=0$
	        \end{tabular}}  (Q1);
	\path [bend angle=40, bend right] (Q2) edge [] node{\begin{tabular}{c}
		$\coulclock{y}=\parami{3}$ \\[1mm]
		\coulact{coffee!} 
		\end{tabular}} (Q0);

\end{tikzpicture}

\caption{An example of a coffee machine}
\label{fig:exPTA}
\end{figure}
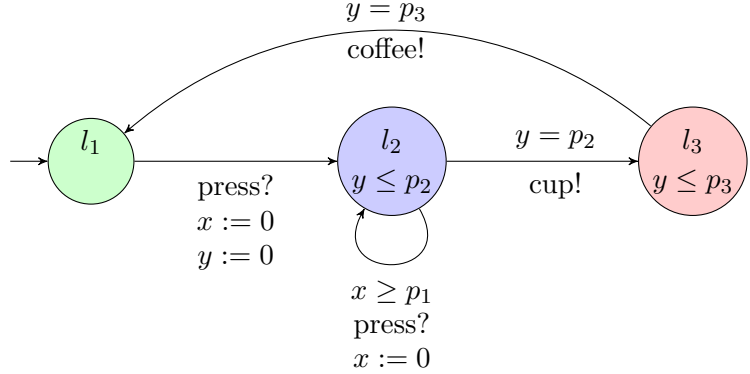

\begin{example}
	The PTA in \cref{fig:exPTA} has three locations, three parameters
	$\param_1$, $\param_2$,~$\param_3$ and two clocks~$x$ and~$y$.
	Both clocks are parametric.
	This~PTA is deterministic.
\end{example}

\subsection{Semantics of Parametric Timed Automata}

Given a PTA $\calA = \tuple{\Sigma, \Loc, \locinit, \Clock, \Param, \invariant, \steps}$, and a parameter valuation~\(\pval\),
$\valuate{\calA}{\pval}$~denotes the automaton obtained from~$\calA$
by substituting every occurrence of a parameter \(\param_i\) by the
constant \(\pval(\param_i)\) in the guards and invariants.  Then
$\valuate{\calA}{\pval}$ is a timed automaton~\cite{AD90}.
The~configurations of a timed automaton are the
pairs~$(\loc,\clockval)$ where $\loc$ is a location and $\clockval$ is
a clock valuation. Moving from one configuration to another one is
allowed depending on the transitions of the PTA, which gives rise to
an infinite-state transition system:
\begin{definition}
  Given a PTA $\calA = \tuple{\Sigma, \Loc, \locinit, \Clock, \Param, \invariant, \steps}$, and a parameter valuation~\(\pval\), the~semantics of
  $\valuate{\calA}{\pval}$ is given by the timed transition system $\tuple{\Q,
  \qinit, \Steps}$ where
  $\Q = \{ (\loc, \clockval) \in \Loc \times
                  (\bbR+)^\Clock \mid
                  \valuate{\valuate{\invariant(\loc)}{\pval}}{\clockval}
                  \text{ evaluates to true} \}$
                  is the set of all valid configurations,
		with initial configuration $\qinit = (\locinit, \mathbf{0}_\Clock) $,
and
$((\loc, \clockval), (d,e), (\loc', \clockval')) \in \Steps$
                  whenever $e$ is a transition
                  $(\loc,\guard,a,\resets,\loc')\in\steps$ such that
                  $\clockval+d\models \valuate\guard\pval$
                  $\clockval'=(\clockval+d)[\resets\mapsto 0]$. 
\end{definition}

A run of a TA is a maximal sequence of consecutive transitions of the timed
transition system associated with the~TA. For~the sake of readability, we
usually write runs as $s_0 \Fleche{d_0,e_0} s_1\Fleche {d_1,e_1} \cdots
\Fleche{d_{m-1},e_{m-1}} s_m\cdots$. With~\emph{maximal}, we~mean that a run
may only be finite if its last configuration has no outgoing transition.
%
The timed word associated to a run $s_0 \Fleche{d_0,e_0} s_1\Fleche {d_1,e_1}
\cdots \Fleche{d_{m-1},e_{m-1}} s_m\cdots$ is the (finite or infinite)
sequence~$(d_i,a_i)_i$ such that for all~$i$, $a_i$ is the action of
edge~$e_i$. The corresponding untimed word is the word~$(a_i)_i$. The timed
(resp.~untimed) language of a TA~$\A$, denoted by $\TLang(\A)$
(resp.~$\ULang(\A)$
\nm{je prefere utiliser ULang, quitte \`a le definir comme Lang}),
is the set of timed (resp.~untimed) words associated with
runs of this automaton. Similarly, the untimed trace associated with
the run $s_0 \Fleche{d_0,e_0} s_1\Fleche {d_1,e_1} \cdots
\Fleche{d_{m-1},e_{m-1}} s_m\cdots$ is the sequence $(l_i,a_i)_i$ s.t. $l_i$
is the location of~$s_i$ and $a_i$ is the action of edge~$e_i$. The~set of
untimed traces of~$\A$ is denoted by~$\UTraces(\A)$.

A~configuration~$s=(\loc, \clockval)$ is said to be reachable in
$\calA$ under valuation~$\pval$
if $s$ belongs to a run of
$\valuate{\calA}{\pval}$; 
a~location~$\loc$ is reachable if some configuration of the
form~$(\loc,\clockval)$~is reachable.

\subsection{Symbolic Semantics of Parametric Timed Automata}

Following, \eg{} \cite{HRSV02,ACEF09,JLR14}, we~now define a symbolic semantics for PTA:
\begin{definition}
	A symbolic state of a PTA~$\calA$ is a pair $(\loc, \C)$ where $\loc
        \in \Loc$ is a location, and $\C$ is a generalized constraint. 
\end{definition}
Given a parameter valuation~$\pval$, a~symbolic state $\state = (\loc, \C)$ is
\emph{$\pval$-compatible} if $\pval \models \project{\C}{\Param}$.
The computation of the symbolic state space relies on the $\Succ$ operation.
The initial symbolic state of~$\calA$ is
$\sinit = (\locinit, \timelapse{(\Clock = 0)}
\land \invariant(\locinit) )$.
Given a symbolic state $s = (\loc, \C)$ and a
transition~$e=(\loc,\guard,\action,\resets,\loc')$, we~let
\[ \Succ_e(s) = \left\{ (\loc', \C') \;\middle|\; \C' = \timelapse{\big(\reset{(\C \land \guard)}{\resets}\big )} \cap \invariant(\loc') \;\land\; C'\not=\emptyset\right\} \]
(notice that this is a singleton or the empty set).
For transitions~$e$ not originating from~$\loc$, we~let $\Succ_e(s)=\emptyset$.
We~then write $\Succ(s) = \bigcup_{e\in \steps} \Succ_e(s)$.
By~extension, given a set~$S$ of symbolic states, $\Succ(S) = \{ s' \mid \exists s \in S \text{ s.t.\ } s' \in \Succ(s)\}$.
Again, this gives rise to an infinite-state transition system, called the
\emph{parametric zone graph} later~on.
A~symbolic run of a PTA from some symbolic state~$s_0$ is a maximal
alternating sequence
$\state_0 \edge_0 \state_1 \edge_1 \cdots $
of symbolic states~$\state_i$ and edges~$\edge_i$ such that $s_{i+1} = \Succ_{e_i}(s_i)$ for all~$i$.
Two~runs are said \emph{equivalent} when they correspond to the same sequences
of edges (hence the same sequences of locations), but may visit different
symbolic states.
From now on, a symbolic run of a PTA~$\A$ refers to a run starting from the initial symbolic state of~$\A$.
By extension, a symbolic state of~$\A$ is a state belonging to a symbolic run of~$\A$.





\subsection{Problems}


%

In this paper, we~address the following two problems:
\begin{definition}
Given a PTA~$\calA$ and a parameter valuation~$\pval$,
\begin{itemize}
\item the \emph{language preservation problem} asks whether there exists another
  parameter valuation~$\pval'$ giving rise to the same untimed language (\ie{}
  such that $\Lang(\valuate{\calA}{\pval}) = \Lang(\valuate{\calA}{\pval'}$);
\item the \emph{trace preservation problem} asks whether there exists another
  parameter valuation~$\pval'$ giving rise
  to the same set of traces (\ie{} such that $\Traces(\valuate{\calA}{\pval}) =
  \Traces(\valuate{\calA}{\pval'}$)~\cite{ACEF09}.
\end{itemize}

The \emph{continuous} versions of those problems additionally require that the
language (resp.\ set of traces) is preserved under any other valuation
of the form $\lambda\cdot \pval + (1-\lambda)\cdot \pval'$, for
$\lambda\in[0,1]$ (with the classical definition of addition and scalar
multiplication). 
\end{definition}





\section{Undecidability of the Preservation Problems in General}\label{section:undecidable}


	


\subsection{Undecidability of the Language Preservation Problem}
\label{ss:language-preservation}

\begin{restatable}{theorem}{thmdisclang}
\label{thm-language}
  The language preservation problem for PTA with one parameter is undecidable
  (both over discrete 
  and continuous time, and for integer and rational parameter valuations).
\end{restatable}


\begin{proof}
  The proof proceeds by a reduction from the halting problem for two-counter
  machines. We~begin with reducing this problem into the classical problem of
  reachability emptiness (``EF-emptiness'') in parametric timed automata,
  namely:
  ``\emph{is the set of valuations of the
    parameters for which the target location is reachable empty?}''
  We~then extend the construction to our original problem.


  \smallskip
  Fix a \emph{deterministic} two-counter machine~$\calM=\tuple{S,T}$:
  such a machine is a finite-state transition system equipped with
  two counters $c_1$ and~$c_2$, initially set to~zero.
  The~transitions of a two-counter machine can be of two different forms:
  \begin{itemize}
  \item from state $s_i$, increment $c_k$ and go to $s_j$. Such a
    transition is denoted by $(s_i;\incr{c_k};s_j)$ in the sequel;
  \item from state $s_i$, if $c_k=0$ then go to~$s_j$, else
    decrement $c_k$ and go to~$s_l$. Such a transition is denoted
    $(s_i; s_j; \decr{c_k}; s_l)$.
\end{itemize}
  In~particular, both counters may only take nonnegative values.

The machine starts in state $s_0$ and halts when it reaches a
particular state~$\shalt{}$. The~halting problem for two-counter
machines is undecidable~\cite{Minsky67}.

We encode the halting problem of two-counter machines into our problem
for PTA.  Given a two-counter machine~$\calM$, we~build a~PTA whose
runs encode the runs of~$\calM$.  Our~PTA uses four clocks: clock~$t$
will serve as a tick (it~will be reset exactly every $\param$ time
units, where $\param$~is the parameter), and we will have a
correspondence between a configuration of the timed automaton and a
configuration of the two-counter machine every time~$t$ is reset;
clocks~$x_1$ and~$x_2$ are used to store the values of counters~$c_1$
and~$c_2$ of~$\calM$, with the correspondence~$x_1=c_1$ and~$x_2=c_2$
when~$t=0$; finally, clock~$z$ is used to count the number of steps of
the two-counter machine that have been simulated during a computation;
this is where our construction differs from the classical ones
(\eg{}~\cite{AHV93,JLR14,BBLS15}), as~we~use the parameter~$\param$ to
bound the length (number of~steps) of the computations of~$\calM$.
Notice that~$\param$ is thus also an upper bound on the values of
both~$c_1$ and~$c_2$.
%

The parametric timed automaton~$\calA$ associated with~$\calM$ is
defined as follows:
\begin{itemize}
\item 
  its set of locations has three copies of the set~$S$ of
  states of~$\calM$: for each~$s\in S$, there is a \emph{main
    location} with the same name~$s$, and two \emph{intermediary
    locations} named~$\mybar s$, $\copyone s$; 
\item each location of~$\calA$ has invariants requiring all clocks to
  never exceed~$\param$; all~intermediary locations carry two self-loops,
  resetting clocks~$x_1$ and $x_2$ when they reach
  value~$\param$. Additionally, locations~$\mybar{s}$ have a self-loop
  resetting~$t$ when it reaches~$\param$, and locations~$\copyone{s}$
  have a self-loop resetting~$z$ when it reaches~$\param$.
  
%

  The rough intuition is as follows: the total time elapsed between
  two consecutive main locations will be~$\param$. If~in the meantime we
  reset each clock exactly when it~reaches value~$\param$ (which is
  the role of the self-loops), then the values encoded by the clocks
  is unchanged. By resetting some clock one time unit earlier or
  later, we~can encode an increment or decrement of the associated
  counter.

  More precisely, each transition in~$\calM$ gives rise to several
  transitions in~$\calA$:
    \begin{itemize}
    \item first, for each main location~$s$,
      there is a transition from~$\copyone s$ to~$s$ guarded with~$t=\param$
      and resetting~$t$, 
      and 
      a transition from~$s$ to~$\mybar s$,
      guarded with~$z=\param-1$ and resetting~$z$; this encodes
      incrementation of~$z$;
    \item then, for transitions of~$\calM$ of the form
      $(s_i;\incr{c_k};s_j)$, there is a transition from $\mybar{s_i}$ to
      $\copyone{s_j}$ guarded with $x_k=\param-1$ and resetting~$x_k$;

      For transitions of the form $(s_i;s_j;\decr{c_k};s_l)$, there
      are two transitions from~$\mybar{s_i}$: one~is guarded with $t=0
      \land x_k=0$, thereby testing that the counter encoded by~$x_k$
      equals~$0$; this transitions goes to
      location~$\copyone{s_j}$. The~second transition\footnote{The
        guard $t\not=1 \land x_k=1$ is not convex, which formally is
        not allowed in our models; but this is easily encoded by
        duplicating the transition.} is guarded with $t\not=1 \land
      x_k=1$; it~resets clock~$x_k$ and goes to~$\copyone{s_l}$.

    \end{itemize}
  \end{itemize}
  
	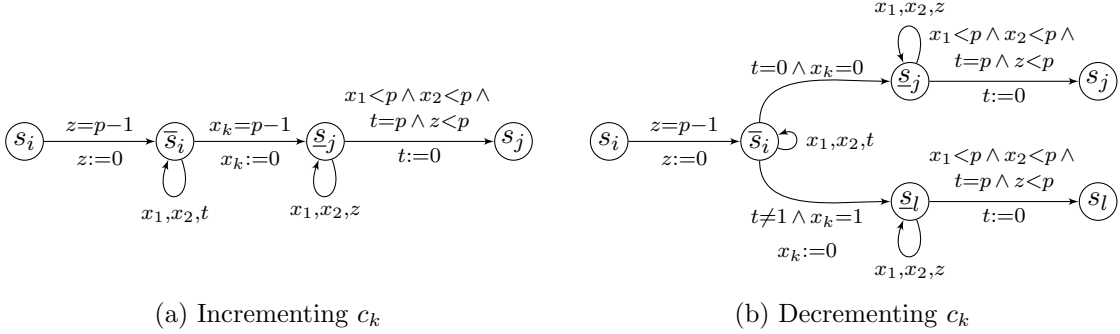
\begin{figure}[ht]
	\centering
    \begin{subfigure}[c]{0.49\textwidth}
      \centering
      \begin{tikzpicture}[-latex',xscale=1]
        
	\node[moyrond] (s) at(0,0) {}; \node at (s)  {$\state_i$};
	\node[moyrond] (sbarre) at(2,0) {}; \node at (sbarre) {$\mybar{\state_i}$};
	\node[moyrond] (s') at(4,0) {}; \node at (s') {$\copyone{\state_j}$};
	\node[moyrond] (sj) at(6.5,0) {}; \node at (sj)  {$\state_j$};
        \path[use as bounding box] (0,-1.9) -- (0,1.5);
                        
        \everymath{\scriptstyle}
	\path
	(sbarre) edge node {\begin{tabular}{c}
	    $x_k=\param-1$\\
	    $x_k := 0$
	\end{tabular}} (s')
	(s) edge node{\begin{tabular}{c}
	    $z = \param-1$\\
	    $z := 0$ 
	\end{tabular}} (sbarre)
        (s') edge node{\begin{tabular}{c}
	    $x_1< \param \et x_2<\param\et{}$\\[-1mm] 
	    $t = \param \et z<\param$\\ 
	    $t := 0$\\\null
        \end{tabular}} (sj)
        (sbarre) edge[loop,out=-70,in=-110,looseness=10]
        node[below] {$x_1,x_2,t$} (sbarre)
        (s') edge[loop,out=-70,in=-110,looseness=10]
        node[below] {$x_1,x_2,z$} (s')
	;
        \node[location, draw=none,fill=none,use as bounding box] at (4,-1) {};
      \end{tikzpicture}
      \caption{Incrementing $c_k$}
      \label{figure:reduction:increment}
    \end{subfigure}%
	\hfill
	\begin{subfigure}[c]{0.49\textwidth}
          \centering
	  \begin{tikzpicture}[-latex',xscale=1]

			\node[moyrond] (s) at(0,0) {}; \node at (s) {$\state_i$};
			\node[moyrond] (sbarre) at(2,0) {}; \node at (sbarre) {$\mybar\state_i$};
			\node[moyrond] (s0) at(4,.8) {}; \node at (s0) {$\copyone{\state_j}$};
			\node[moyrond] (ms0) at(6.5,.8) {}; \node at (ms0) {$\state_j$};
			\node[moyrond] (s1) at(4,-.8) {}; \node at (s1) {$\copyone{\state_l}$};
			\node[moyrond] (ms1) at(6.5,-.8) {}; \node at (ms1) {$\state_l$};
                        \path[use as bounding box] (0,-1.9) -- (0,1.5);

        \everymath{\scriptstyle}
	\path
	(sbarre) edge[out=90,in=180] node[above left=-7pt,pos=.9] {\begin{tabular}{c}
	    $t = 0 \et x_k = 0$
	\end{tabular}} (s0)
	(sbarre) edge[out=-90,in=180] node[below left=-7pt,pos=.9] {\begin{tabular}{c}
	    $t\not=1\et x_k = 1$\\
            $x_k := 0$ 
	\end{tabular}} (s1)
	(s) edge node {\begin{tabular}{c}
	    $z = \param-1$\\
	    $z := 0$
	\end{tabular}} (sbarre)
        (s0) edge node{\begin{tabular}{c}
	    $x_1< \param \et x_2<\param\et{}$\\ [-1mm]
	    $t = \param \et z<\param$\\ 
	    $t := 0$\\\null
        \end{tabular}} (ms0)
        (s1) edge node{\begin{tabular}{c}
	    $x_1< \param \et x_2<\param\et{}$\\ [-1mm]
	    $t = \param \et z<\param$\\ 
	    $t := 0$\\\null
        \end{tabular}} (ms1)
        (sbarre) edge[loop,out=-20,in=20,looseness=5]
        node[right] {$x_1,x_2,t$} (sbarre)
        (s0) edge[loop,out=70,in=110,looseness=10]
        node[above] {$x_1,x_2,z$} (s0)
        (s1) edge[loop,out=-70,in=-110,looseness=10]
        node[below] {$x_1,x_2,z$} (s1)
	;
	  \end{tikzpicture}
	  \caption{Decrementing $c_k$}
	  \label{figure:reduction:decrement}
	\end{subfigure}%

	\caption{Encoding a 2-counter machine. The lists of clocks on
          self-loops indicate which clocks are reset when they reach
          value~$\param$.}
	\label{figure:reduction}
	\end{figure}

Correctness of this construction is expressed as follows:
\begin{restatable}{lemma}{lemmareduc}
  \label{lemma:reduction}
  The two-counter machine~$\calM$ has a halting computation
  from~$(s_0,({c_1=0},$ ${c_2=0}))$ if, 
  and only~if, there is a run in~$\pval(\calA)$ reaching the corresponding
  location~$\shalt{}$ from the initial
  configuration~$(s_0,(t=0,x_1=0,x_2=0,z=0))$ in $\valuate{\A}{\pval}$
  with $\pval(\param)>0$.

  Moreover, if~$\calM$ has no halting computations, then for any
  $\pval$, $\pval(\calA)$ eventually reaches a deadlock.
\end{restatable}


\begin{proof}
  We~prove that the modules for incrementing and decrementing counters
  correctly implement these operations, as long as~$z$ is small enough.
  Precisely:
  \begin{itemize}
  \item assume that the PTA is in configuration
    $(\state_i,(t=0,x_1=x_1^0,x_2=x_2^0,z=z^0))$ when entering the
    module encoding transition $(\state_i,\incr{c_1},\state_j)$
    incrementing~$c_1$.  Assume that $\max(x_1^0,x_2^0)\leq
    z^0<\param$ (and~$t=0$) when entering that module (which is true
    initially, and we will prove is preserved when entering the next
    module).
    Because clock~$z$ cannot be reset in~$\mybar{\state_i}$, at most
    $\param$ time unit can elapse in that location; similarly
    in~$\copyone{\state_j}$. Hence either the automaton
    reaches location~$\state_j$, or it ends up in a deadlock.
    If~the automaton reaches~$\state_j$, then the total
    time elapsed along this run will be either~$\param$ or~$2\cdot\param$,
    because clock~$t$ is initially~zero and it
    equals~$\param$ at the end of the run; it~may have been reset once in
    the meantime in~$\mybar{\state_i}$, when its value was exactly~$\param$.
    Each resetting self-loop amounts to decrementing
    the value of its associated clock by~$\param$, since it tests if
    the clock equals~$\param$ and resets it~to~zero. Hence the final
    value of~$x_2$ is unchanged. The~transition from~$\state_i$
    to~$\mybar{\state_i}$ amounts to decreasing clock~$z$ by~$\param-1$;
    all other transitions preserve the value of clock~$z$
    modulo~$\param$, so that in the end the value of~$z$ is augmented
    by~$1$. The same argument applies to~$x_1$. In~the end, if the
    module is eventually exited, the automaton reaches configuration
    $(\state_j,(t=0,x_1=x_1^0+1, x_2=x_2^0,z=z^0+1))$, as expected.
    
    Conversely, assuming that $\max(x_1^0,x_2^0)\leq z^0\leq\param-1$
    (and~$t=0$) when entering that module, then there is a path from
    $(\state_i,(t=0,x_1=x_1^0,x_2=x_2^0,z=z^0))$ to
    $(\state_j,(t=0,x_1=x_1^0+1, x_2=x_2^0,z=z^0+1))$.
    If~$z^0>\param-1$, then $\state_i$ can not
    be exited, and the automaton ends in a deadlock.
    


    The case of incrementation of~$c_2$ is symmetric.
    
  \item similarly, assume that the PTA is in configuration
    $(\state_i,(t=0,x_1=x_1^0,x_2=x_2^0,z=z^0))$ when entering the
    module decrementing~$c_1$. Then~counter~$c_1$ equals~zero if,
    and only~if, it~holds $x_1^0=0$ when~$t=0$. In~that case, the
    automaton can only proceed to~$\state_j$: if~that location is ever
    reached, again the total time elapsed will be an integer multiple
    of~$\param$, and by a similar analysis as above, we~get that the
    automaton will reach configuration $(\state_j,(t=0,x_1=x_1^0=0,
    x_2=x_2^0,z=z^0+1))$.

    On the other hand, if~counter~$c_1$ is not~zero, \ie, if $x_1$ is
    not~zero (nor~$\param$) when~$t=0$, then the automaton can only
    reach location~$\state_l$. With a similar argument as above,
    we~obtain that the automaton will then reach configuration
    $(\state_l,(t=0,x_1=x_1^0-1,x_2=x_2^0,z=z^0+1))$, as required.

    Conversely, if $\max(x_1^0,x_2^0)\leq z^0\leq\param-1$ (and~$t=0$)
    when entering this module, then in both cases ($c_1=0$ and
    $c_1>0$) there is a path to the corresponding exit configuration
    in that module.
    If on the other hand $z^0>\param-1$, then $\state_i$ can not be left.
    Finally, the case of decrementation of~$c_2$ is    symmetric.
  \end{itemize}

  From these results, we~obtain the fact that if there is a run to
  location~$\shalt$ (for some value of~$\param$), then it corresponds
  to a valid halting run of the two-counter machine; conversely, if
  the two-counter machine has a halting run of length~$n$, then for
  $\pval(\param)=n+1$, we~can build a run in~$\pval(\calA)$ reaching
  location~$\shalt$. Finally, if the two-counter machine has no
  halting computation, then eventually the value of clock~$z$ will
  exceed~$\param-1$ when entering a module, which will result in a deadlock.
\end{proof}

  \smallskip

	We~now explain how to adapt this construction to the language
        preservation problem. The idea is depicted on
        \cref{figure:language-preservation} (where all transitions are
        labeled with the same letter~$a$): when $\pval(\param)=0$, the
        automaton accepts the untimed language $\{a^\omega\}$. Notice
        that the guard $\param=0$ in the automaton can be encoded by
        requiring $t=0\et t=\param$. On the other hand, when
        $\pval(\param)>0$, we have to enter the main part of the
        automaton~$\calA$, and mimic the two-counter machine. From our
        construction above, if the run of the two-counter machine is
        halting run, then for some value~$\pval(\param)$,
        location~$\shalt$, and then~$s_\infty$, will be reached, and
        the untimed language will be the same as
        when~$\pval(\param)=0$.  Conversely, if the two-counter
        machine does not halt, then for any value of~$\pval(\param)$,
        the automaton will reach a deadlock, and it will not accept~$a^\omega$.

  Finally, notice that our reduction is readily adapted to the discrete-time
  setting, and\slash or to integer-valued parameters.
\end{proof}

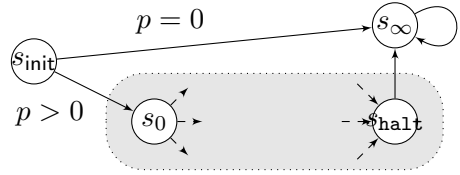
\begin{figure}[tb]
	\centering
	\begin{tikzpicture}
		\begin{scope}[scale=.8]
			\draw (0,0) node[circle,draw,minimum height=6mm] (init) {} node {$\sinit$};
			\draw[fill=black!10!white,rounded corners=5mm,dotted] (1.2,-1) |-
				(6.8,-1.8) |- (1.2,-.2) -- (1.2,-1);
			\draw (2,-1) node[fill=white,circle,draw,minimum height=6mm] (0) {} node {$s_0$};
			\draw (6,-1) node[fill=white,circle,draw,minimum height=6mm] (halt) {} node {$\shalt{}$};
			\draw (6,.6) node[circle,draw,minimum height=6mm] (infty) {} node {$s_\infty$};
			\draw[-latex'] (init) -- (infty) node[midway,above left] {$\param=0$};
			\draw[-latex'] (init) -- (0) node[midway,below left] {$\param>0$};
			\draw[dashed,-latex'] (0) -- +(45:8mm);
			\draw[dashed,-latex'] (0) -- +( 0:8mm);
			\draw[dashed,-latex'] (0) -- +(-45:8mm);
			\draw[dashed,latex'-] (halt) -- +(45:-9mm);
			\draw[dashed,latex'-] (halt) -- +( 0:-9mm);
			\draw[dashed,latex'-] (halt) -- +(-45:-9mm);
			\draw[-latex', rounded corners=3mm] (halt) -- (infty);
			\draw[-latex'] (infty) edge[-latex',loop,out=30,in=-30,looseness=7] (infty);
		\end{scope}
	\end{tikzpicture}
	\caption{Encoding the halting problem into the language-preservation problem}
	\label{figure:language-preservation}
\end{figure}

\begin{remark}\label{remark:p-1}
Our construction uses both $\param$ and $\param-1$ in the clock
constraints, as well as parametric constraints~$p=0$ and $p>0$.  This
was not allowed in~\cite{AHV93} (where three different parameters were
needed to compare the clocks with~$\param$, $\param-1$
and~$\param+1$). Our construction could be adapted to only allow
comparisons with $\param-1$ (hence to use only one parameter), while
keeping the number of clocks unchanged:
\begin{itemize}
\item the parametric constraints $p=0$ and $p>0$ could be respectively
  encoded as $(t=p) \et (t=0)$ and $(t<p) \et (t=0)$;
\item transitions guarded by
$x=\param$ (which always reset the corresponding clock~$x$) would then be
encoded by a first transition with~$x=\param-1$ resetting~$x$ and moving to a
copy of~$\calA$ where we \emph{remember} that the value of~$x$ should be
shifted up by~$\param-1$. All locations have invariant $x\leq 1$, and
transitions guarded with~$x=1$, resetting~$x$ and returning to the main copy
of~$\calA$. The same can be achieved for the other clocks, even if it means
duplicating~$\calA$ several times (twice for each clock).
\end{itemize}
\end{remark}





Let us now show that this 
undecidability result is robust \wrt{} some variations in the definition of the problem.

\begin{proposition}\label{proposition:inclusion}
  Given a PTA~$\A$ and a parameter valuation~$\pval$, the existence of
  a valuation~$\pval' \neq \pval$ such that the 
  language of $\valuate{\A}{\pval'}$ is a strict subset of
that of~$\valuate{\A}{\pval}$ is undecidable (similarly for
non-strict subset, and for strict and non-strict superset).
\end{proposition}
\begin{proof}
We show that all four problems are undecidable:
\begin{itemize}[align=left]
\item[($\subseteq$)] The result follows directly from the encoding in
  \cref{figure:language-preservation}: if the two-counter machine
  halts, then the untimed language is $a^\omega$ for some positive
  value of~$\param$, as well as for~$\param=0$. If it does not halt,
  then the untimed language is made of finite words only
  when~$\param>0$.

\item[($\subsetneq$)] Consider again the encoding in
  \cref{figure:language-preservation}: add a transition from the
  initial location to a new location guarded with $\param = 0$, and
  this time labeled with action~$b$ (recall that all other transitions
  are labeled with~$a$).

  With this new transition, the untimed language for $\param = 0$
  becomes $\{a^\omega , b\}$.
  Now, if the two-counter machine halts, for some positive value of
  the parameter, the untimed language of the automaton
  is~$\{a^\omega\}$, which is a strict subset of the language of the
  automaton for~$\param=0$.
  On the other hand, if the machine does not halt, then for any
  positive parameter valuation, the automaton reaches a deadlock,
  hence its untimed language is a (non-empty) set of finite words
  in~$a^+$ (recall that $a^+$ denotes the set of all words made of an
  arbitrary number of~$a$s strictly greater than~0), and it is not a
  subset of~$\{a^\omega,b\}$.
\item[($\supseteq$)] Same argument as for $\subseteq$.
\item[($\supsetneq$)] We use a reasoning dual to the~$\subsetneq$ case:
  add a transition from the initial location to a new location guarded
  with $\param > 0$, labeled with action~$b$.

  Then for~$\param=0$, the untimed language still is~$\{a^\omega\}$.
  Then if the two-counter machine halts, then for some positive value
  of~$\param$, the untimed language is $\{a^\omega,b\}$; if the
  machine does not halt, the~language contains a finite word in~$a^+$,
  whatever the (positive) value of~$\param$.\qedhere
\end{itemize}
\end{proof}



We considered so far a definition of the untimed language as the set
of untimed words associated with maximal runs, \ie{} runs that are
either infinite or blocking.
An~alternative definition of the untimed language could be the set of
untimed words associated with all finite runs (non-necessarily
maximal); note that this definition yields an untimed language that is
prefix-closed.
We~prove that all results above (\ie{} \cref{thm-language,proposition:inclusion}) extend to this alternative definition.
We~first consider the equality of language, and then the four
variations of the problem, with (strict) inclusion instead of equality
of the set of untimed words.

\begin{proposition}\label{proposition:prefix-closed}
	Given a PTA~$\A$ and a parameter valuation~$\pval$, the
        problems of the existence of a valuation~$\pval' \neq \pval$
        such that the set of non-necessarily maximal finite untimed
        words of $\valuate{\A}{\pval'}$ is equal to, strictly included
        in, included in or equal to, larger than or equal to, or
        strictly larger than that of~$\valuate{\A}{\pval}$,
        respectively, are undecidable.
\end{proposition}
\begin{proof}
  We~begin with handling untimed-language  equality,
  again relying on the encoding in \cref{figure:language-preservation}.
  For $\param = 0$, the untimed language is~$a^+$.  For $\param > 0$
  if the two-counter machine does not halt, recall that any run will
  eventually reach a deadlock (for any positive value of~$\param$),
  yielding $a^{\leq N} = \{a^k \mid 0\leq k\leq N\}$ (for some
  strictly positive $N \in \grandn$ depending on the value
  of~$\param$) as the untimed language.
  On~the other~hand, if~the machine halts, then for some value
  of~$\param$, the automaton has an infinite run, and the untimed
  language is~$a^+$.
%

  This gives that there exists a parameter valuation $\param > 0$ with
  the same set of non-necessarily maximal finite untimed words as for
  $\param = 0$ if, and only~if, the two-counter machine halts.

  We now prove the results for inclusion relations:
  \begin{itemize}[align=left]
  \item[($\supseteq$)] The argument for language equality above
    also applies in this case.
    
  \item[($\subseteq$)] We modify the construction to prove this case:
    from each main location~$\state_i$ of the automaton, we~add a
    transition labeled with~$b$ and guarded by $t=0\et z=\param$: this
    transition can only be taken after $\param$ steps of the
    two-counter machine have been simulated. Then if the two-counter
    machine does not halt, for any positive value of~$\param$,
    the~language contains at least one word that contains a~$b$, hence
    it cannot be included in the language for~$\param=0$; on the other
    hand, it the two-counter machine halts, then for some positive
    value of~$\param$ both languages are the same (hence the inclusion
    holds).

  \item[($\subsetneq$)] As for the same case in the proof of
    Prop.~\ref{proposition:inclusion}, it~suffices to modify the
    automaton in order to add an extra word (\eg~$b$) when~$\param=0$.


		
  \item [($\supseteq$)] Similarly, it~suffices to add one word to the
    language for any positive valuation of~$\param$.\qedhere
%
%
%
	\end{itemize}
\end{proof}

\subsection{Undecidability of the Trace Preservation Problem.}
\label{sec-trace}






In this section, we provide 
two
proofs of the following result:
\begin{restatable}{theorem}{thmdisctrace}
\label{theorem:trace-preservation-undec}
  The trace-preservation problem for PTA with one parameter is
  undecidable.
\end{restatable}

%
We~propose two different proofs of this result:
\begin{enumerate}
\item the first proof (\cref{ss:one-location}) is by a generic
  transformation of (parametric) timed automata without zero-delay
  cycle into one-location timed automata; the~transformation~involves diagonal
  constraints, uses an unbounded number of clocks,
  but does not increase the number of parametric clocks;
\item the second proof (\cref{ss:bounded-locations}) does not involve
  diagonal constraints. It~involves eight locations, but
  with an unbounded number of transitions and an unbounded number of
  parametric clocks.
\end{enumerate}

\subsubsection{Encoding timed automata into one-location timed automata.}\label{ss:one-location}

Our first proof relies on the encoding of TA (with the restriction that no
sequence of more than $k$ transitions may occur in zero delay, for some~$k$;
equivalently, those timed automata may not contain zero-delay cycles) into an
equivalent TA with a single location; this reduction uses $k \times |\Loc|$
additional clocks (where $|\Loc|$ denotes the number of locations of~$\A$) and
requires diagonal constraints, \ie{} constraints comparing clocks with each
other (of the form $\clock_1 - \clock_2 \compOp c$).

This result extends to PTA (provided that $k$ does not depend on the
value of the parameters), and the additional clocks are
non-parametric.
Using this reduction, the undecidability of the language preservation
(\Cref{thm-language}) trivially extends to trace preservation.
Let us first show the generic result for~TA.

\begin{restatable}{proposition}{proponeloc}
\label{proposition:TA-collapse-locations}
  Let $\calA$ be a TA in which any run starts with a positive delay,
  and such that for some~$k$, no sequence of more than $k$
  transitions can occur in zero delay. Then there exists an equivalent
  TA~$\calA'$ with only one location and $k \times |\calA|+2$
  additional clocks, such that the timed languages of~$\calA$
  and~$\calA'$ are the same.
\end{restatable}
\begin{proof}
\looseness=-1
  We~begin with the intuition behind our construction: 
  each location~$\ell$ of the automaton~$\calA$ is encoded using
  an extra clock~$\clock_\ell$, with the following property:
  when location~$\ell$ is entered in the original automaton,
  the~associated clock~$\clock_\ell$ is reset in the one-location automaton.
  An~extra clock~$x_0$ is reset along each transition.
%
  This~way, when the automaton is visiting~$\ell$, it~holds
  $\clock_\ell - \clock_0 = 0$. However, the converse does not hold,
  because several transitions may be taken in zero delay.

  To overcome this difficulty, we~use $k+1$ copies of~$\clock_\ell$,
  numbered $x^0_\ell$ to~$x^{k}_{\ell}$. In~the encoding, each
  transition $t=(\ell,g,a,R,\ell')$ is encoded as several self-loops
  on the single location of~$\calA'$: for each $1\leq i\leq k$, one
  self-loop encodes the effect of taking transition~$t$ as the $i$-th
  transition in a sequence of zero-delay transitions; additionally,
  one transition encodes the effect of taking~$t$ right after a
  positive delay. Formally:
  \begin{itemize}

  \item for each $1\leq i\leq k$, one self-loop
    is guarded with the
    conjunction of~$g$ and
    \begin{equation*}
   \clock_0=0 \et 
    \left[ x^{i-1}_{\ell}=0 \et \bigwedge_{\ell''\in L}
      x^{i}_{\ell''}>0 \right];
    \end{equation*}
    The first part of the latter constraint imposes that $\clock_0=0$,
    hence no delay may have elapsed since the previous transition. The
    second part of the constraint characterizes that a sequence of
    exactly $i$ zero-delay transitions has been taken, and has reached
    location~$\ell$. In~order to encode the effect of transition~$t$,
    the self-loop carrying this guard is labeled with~$a$, and resets
    the clocks in~$R$ and $x^{i}_{\ell'}$ (no~need to reset~$x_0$ as
    it is already~zero).
    
  \item the last self-loop corresponds to transition~$t$ right
    after a positive delay; it~is guarded with the conjunction of the
    guard~$g$ and of the constraint
    \begin{equation*}
    \clock_0>0 \et
    \bigvee_{i\geq 0} \left[ x^i_{\ell}-\clock_0=0 \et \bigwedge_{\ell''\in L}
      x^{i+1}_{\ell''}-\clock_0>0 \right];
    \end{equation*}
    Indeed, after a sequence of $i$ zero-delay transitions (preceded by a
    non-zero-delay transition), it~holds $x^{i+1}_\ell-\clock_0>0$ for
    all~$\ell$, and only the location~$\ell$ reached at the end of the sequence
    satisfies $x^i_\ell=\clock_0$.
    This self-loop is labeled with~$a$, and resets the clocks in~$R$ as well as
    $x_0$ and~$x^0_{\ell'}$.

  \end{itemize}
  We~require, for the time being, that initially all clocks have
  positive values, except for~$x_0$ and~$x_{\ell_0}^0$. In~that case,
  there is a one-to-one correspondence between runs in~$\calA$ (never
  involving more than $k$ consecutive transitions in zero delay) and
  those in~$\calA'$, so that both automata accept the same timed language.
  
  Setting a special initial configuration is required for our encoding
  to be correct.  The~extra requirement that any run in~$\calA$ has to
  begin with a positive delay allows us to circumvent this problem:
  we~add an extra clock~$x_1$, which will never be reset (hence
  we~have $x_0=x_1$ only for the first transition); all transitions
  $(\ell_0,g,a,R,\ell')$ from the initial location of the original
  automaton then give rise to a self-loop in the one-location
  automaton, guarded with $g\et x_0>0\et x_0=x_1$, labeled with~$a$
  and resetting $x_0$, $x_{\ell'}^0$, and the clocks in~$R$.
\end{proof}


The above transformation can obviously be applied to PTA,
with the property that the
timed language is preserved for any valuation of the parameters.
\Cref{proposition:TA-collapse-locations} can be extended to~PTA as follows:


\begin{proposition}\label{proposition:PTA-collapse-locations}
  Let $\calA$ be a PTA for which there exists an integer~$k$ such that,
  for any parameter value,
  all runs start with a positive
  delay, and  no sequence of more than
        $k$ transitions occurs in zero delay.
	Then there exists an equivalent PTA~$\calA'$ with only one
  location and $k \times |\calA|+2$ additional clocks
  such that for any parameter valuation~$\pval$, the~timed languages of
  $\valuate{\calA}{\pval}$ and $\valuate{\calA'}{\pval}$ coincide. 
\end{proposition}


We~slightly modify the PTA built in the proof of
Theorem~\ref{thm-language} (see
Fig.~\ref{figure:language-preservation}) so that we can apply the
transformation above: for this, we~add a new location before~$\sinit$
(so as to enforce a positive initial delay), and we~constrain the
self-loop on~$s_\infty$ so that some time elapses between consecutive
occurrences.

Now, the resulting one-location PTA uses only one letter, so that its
untimed language corresponds to its set of untimed traces. This proves
our result.


\smallskip
As a remark, let us show that in the general case, deciding whether a
given PTA contains no reachable zero-delay cycles, for some valuation
of the parameters, is undecidable.

\begin{theorem}\label{theorem:zerocycle}
    The existence of a parameter valuation $\pval$ in a PTA $\calA$ such that $\valuate{\calA}{\pval}$ contains no reachable zero-delay cycle is undecidable.
\end{theorem}

\begin{proof}
    Consider the two-counter machine encoding for the EF-emptiness
    problem in the proof of Theorem~\ref{thm-language}. It relies on
    the fact that the values of the counters are encoded modulo
    $\param$, and that we can always find a big enough value of
    $\param$ to correctly encode an halting execution. We can
    therefore exclude values $0$ and $1$ from the possible values of
    $\param$ without changing anything in the proof. To do so we need
    only change the initial location to a fresh one and a transition
    form the new initial location to the former with guard $\clock =
    \param \wedge \clock\geq 2$, that resets all the clocks.

    Now, when $\param\geq 2$, it is easily seen that all gadgets take
    at least $1$ time unit to be traversed.
%
    So, if we add a self-loop with guard \emph{true} to the location
    encoding the halting state, then there exists a parameter
    valuation $v$ such that there is a reachable zero-delay cycle in
    $\valuate{\A}{\pval}$ if, and only~if,
    the two-counter machine halts.
\end{proof}

Knowing whether a PTA contains a zero-delay cycle for all parameter
valuations is also undecidable: the construction above can be lifted
to the undecidability proof for EF-universality found in
\cite[Theorem~7]{ALR16ICFEM}.



\subsubsection{Proof with bounded number of locations}
\label{ss:bounded-locations}

We propose a second proof, where we avoid the use of diagonal
constraints, at the expense of using unboundedly many parametric
clocks. This~proof follows the reduction of the proof of
Theorem~\ref{thm-language}, but with only eight locations: one location
is used to initialize the computation, and the other seven locations
are then visited iteratively, in order to first update the information
about the counters and then about the state of the two-counter
machine. The location of the machine is then stored using as many
clocks as the number of locations of the machine: the~clock with least
value (less than or equal to~$\param$) corresponds to the current location.

From a deterministic two-counter machine~$\calM$ with $n$ states, we
build a PTA with $n+4$ (parametric) clocks: $n$~clocks~$q_1$ to~$q_n$ are used
to store the current location of~$\calM$ (the~only clock with value less than
or equal to the value of the parameter~$\param$ corresponds to the current
state of~$\calM$), two clocks~$\clock_1$ and~$\clock_2$ store the values of
the two counters, clock~$t$ measures periods of $\param$ time~units, and
an~extra clock~$r$ stores temporary information along the~run.
Intuitively, the~PTA cycles between two main locations~$s_1$
and~$t_1$, each round in the cycle encoding the application of a
transition of the two-counter machine (see~\cref{fig-7loc}):
it~goes from~$s_1$ to $t_1$ for updating the values of the counters,
and from~$t_1$ back to~$s_2$ for updating the clock encoding the new
location of~$\calM$.

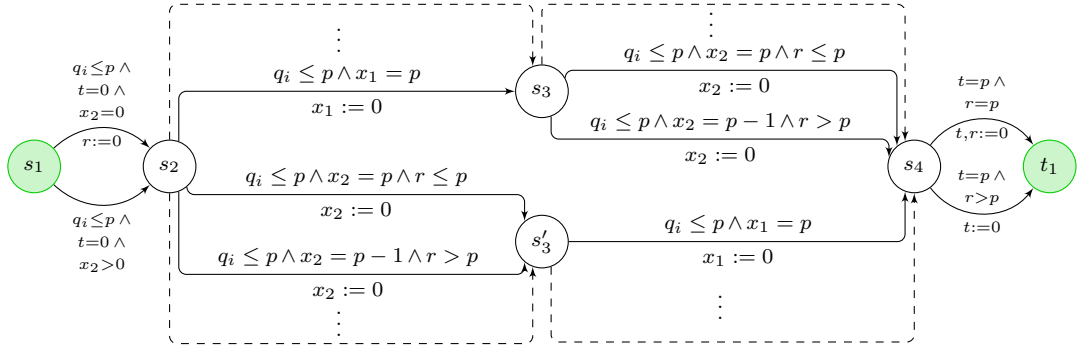
\begin{figure}[htb]
\centering
\scriptsize
\begin{tikzpicture}[xscale=.9]
\path[use as bounding box] (3,-2.4) -- (15.5,2.5);
\draw (2,0) node[rond,vert] (b) {} node {$s_1$};
\draw (4,0) node[rond] (c) {} node {$s_2$};
\draw (9.5,1) node[rond] (d) {} node {$s_3$};
\draw (9.5,-1) node[rond] (d') {} node {$s'_3$};
\draw (15,0) node[rond] (e) {} node {$s_4$};
\draw (17,0) node[rond,vert] (f) {} node {$t_1$};
\draw[-latex'] (b) edge[bend right=40] node[midway,below] {$\substack{q_i\leq
    p \et\\[1mm] t=0\et\\[1mm] \clock_2>0}$} (c); 
\draw[-latex'] (b) edge[bend left=40] node[midway,above] {$\substack{q_i\leq p
    \et \\[1mm] t=0\et \\[1mm] \clock_2=0}$} node[below] {$\substack{r:=0}$} (c);

\draw[-latex',rounded corners=1mm] (c.70) |-
node[above,pos=.75] {$q_i\leq
  p \et \clock_1=p$} node[below,pos=.75] {$\clock_1:=0$}  (d.-180);
\draw[-latex',rounded corners=1mm] (d'.0) -- +(2.5,0) node[above] {$q_i\leq
  p \et \clock_1=p$} node[below] {$\clock_1:=0$} -| (e.-110);

\draw[-latex',rounded corners=1mm] (c.-50) |- +(2.5,-.1) node[above,pos=1] {$q_i\leq
  p \et \clock_2=p \et r\leq p$} node[below] {$\clock_2:=0$} -| (d'.130);
\draw[-latex',rounded corners=1mm] (d.30) |- +(2.5,.1) node[above] {$q_i\leq
  p \et \clock_2=p \et r\leq p$} node[below] {$\clock_2:=0$} -| (e.130);

\draw[-latex',rounded corners=1mm] (c.-70) |- +(2.5,-1.1) node[above]
     {$q_i\leq p \et \clock_2=p-1 \et r>p$} node[below] {$\clock_2:=0$} -|
     (d'.-130); 
\draw[-latex',rounded corners=1mm] (d.-70) |- +(2.5,-.3) node[above]
     {$q_i\leq p \et \clock_2=p-1 \et r>p$} node[below] {$\clock_2:=0$} -|
     (e.160); 

\draw[-latex',rounded corners=1mm,dashed] (c.-90) |- +(2.5,-2) node[above]
     {$\vdots$} -| (d'.-110); 
\draw[-latex',rounded corners=1mm,dashed] (d.90) |- +(2.5,.8) node[below=-2mm]
     {$\vdots$} -| (e.110); 

\draw[-latex',rounded corners=1mm,dashed] (c.90) |- +(2.5,1.8) node[below]
     {$\vdots$} -| (d.110); 
\draw[-latex',rounded corners=1mm,dashed] (d'.-70) |- +(2.5,-1) node[above=2mm]
     {$\vdots$} -| (e.-90);






\draw[-latex'] (e) edge[bend left=50] node[above] {$\substack{t=p\et\\[1mm] r=p}$}
  node[below] {$\substack{t,r:=0}$} (f);
\draw[-latex'] (e) edge[bend left=-50] node[above] {$\substack{t=p\et\\[1mm] r>p}$}
  node[below] {$\substack{t:=0}$} (f);
\end{tikzpicture}
\caption{Encoding a two-counter machine (only one transition
  from~$q_i$, testing and possibly decrementing clock~$\clock_2$, has
  been encoded;
%
  the other transitions would add more transitions between $s_2$
  and~$s_4$).  The~module has two branches, going either through~$s_3$
  or~$s'_3$, depending on the relative values of~$x_1$
  and~$x_2$. Clock~$r$ is used to keep track whether the counter
  is~zero.}\label{fig-7loc}
\end{figure}

More precisely, after spending $\param+1$ time units in the initial
location of the PTA (not displayed on \cref{fig-7loc}), we~take a
transition resetting clocks~$q_1$, $\clock_1$, $\clock_2$, $r$
and~$t$. This sets the initial configuration for starting the
simulation of the two-counter machine. The~PTA then cycles between two
modules. The first module (depicted on \cref{fig-7loc}) is used to
update the values of the clocks encoding the counters: depending on
the instruction to perform (which only depends on the state of the
two-counter machine, as it is deterministic), it~first tests
(between~$s_1$ and~$s_2$) whether the clock ($\clock_1$ or~$\clock_2$)
encoding the counter to be updated by the transition
is~zero, and resets clock~$r$ if needed (in~order to remind that piece of information).
It~then
updates the clock
depending on the transition
of~$\calM$. It~also has to reset the other, non-updated clock when
it~reaches~$\param$; this may occur before or after the reset of the
clock being updated, hence the two branches in the module.
Finally, the module has a transition to its last location~$t_1$, available when~$t=\param$.

From~location~$t_1$, a second module updates the values
of~clocks~$q_i$ depending on the transition to be performed
in~$\calM$. It~suffices to reset the clock~$q_j$ corresponding to the
new location (while~$t=0$, and using the value of clocks~$q_i$ and~$r$
to get the next location to be visited; notice that this gives sufficient
information since we assume that our two-counter machine is
deterministic). The~automaton then returns to~$s_1$ after
letting~$\param$ time unit elapse, and resetting~$t$ and $q_j$ whose
values equal~$\param$.

This is a direct encoding of a two-counter machine as a~PTA. It~can easily be
adapted to follow the reduction scheme of Theorem~\ref{thm-language}, which
entails our result.
Notice that by adding two extra clocks and two
intermediary locations, we~can get rid of comparisons with~$p-1$ and~$p+1$, in
order to use only constraints of the form~$x\sim p$ (see Remark~\ref{remark:p-1}).

\subsection{Undecidability of the Robust Language-Preservation Problem}

The robust language-preservation problem extends the discrete one by
additionally requiring that the language is preserved on a ``line'' of
valuations originating from the reference valuation. This is not the
case of our previous proofs, which require a single parameter
valuation for the reduction to be correct.
In~this section, we~depart from the ``discrete'' setting of the
previous section, and use rational-valued parameters and the full power of
real-valued clocks.


\begin{restatable}{theorem}{thmroblang}
  \label{theorem:robust-language-preservation:undecidable:bPTA}
  The robust language-preservation problem for PTA with one (possibly bounded)
  parameter is undecidable.
\end{restatable}

\begin{proof}
  We~begin by recalling from~\cite{ALR16ICFEM} a reduction
  of the halting problem for counter machines to the EF-emptiness
  problem for 1-parameter~PTA. The~proof is~then adapted to the
  language-preservation problem in the same way as for the proof of
  \cref{theorem:trace-preservation-undec}.

  The encoding of the two-counter machine is as follows: it~uses one
  rational-valued parameter~$p$, one clock~$t$ to tick every time unit, and
  two parametric clocks~$x_1$ and~$x_2$
  for storing the values of the counters~$c_1$ and~$c_2$, with
  $x_i=1-p\cdot c_i$ whenever $t=0$. 

  An initial transition is used to initialize the values of~$x_1$ and~$x_2$
  to~$1$, while it sets~$t$ to~zero. It~also checks that the value of~$p$ is
  in~$(0,1)$. Zero-tests are easily encoded by checking whether~$x_i=1$
  while~$t=0$. Incrementation is achieved by resetting clock~$x_i$ when
  it~reaches $1+p$, while the other clocks are reset when they reach~$1$
  (see~\cref{fig-incr-robust}).
  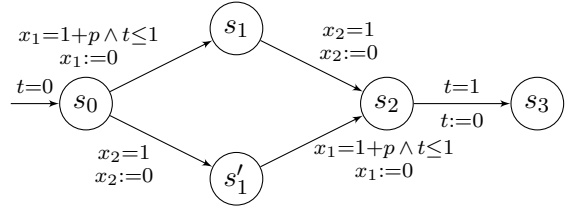
\begin{figure}[tb]
    \centering
    \begin{tikzpicture}
      \draw (0,0) node[rond] (a) {} node {$s_0$};
      \draw (2,1) node[rond] (b) {} node {$s_1$};
      \draw (2,-1) node[rond] (c) {} node {$s'_1$};
      \draw (4,0) node[rond] (d) {} node {$s_2$};
      \draw (6,0) node[rond] (e) {} node {$s_3$};
      \draw[-latex'] (a) -- (b) node[midway,above left=-6pt]
        {$\genfrac{}{}{0pt}{1}{x_1=1+p\et t\leq 1}{x_1:=0}$};
      \draw[-latex'] (a) -- (c) node[midway,below left=-3pt]
        {$\genfrac{}{}{0pt}{1}{x_2=1}{x_2:=0}$};
      \draw[-latex'] (c) -- (d) node[midway,below right=-6pt]
        {$\genfrac{}{}{0pt}{1}{x_1=1+p\et t\leq 1}{x_1:=0}$};
      \draw[-latex'] (b) -- (d) node[midway,above right=-3pt] 
        {$\genfrac{}{}{0pt}{1}{x_2=1}{x_2:=0}$};
      \draw[-latex'] (d) -- (e) node[midway,above] {$\scriptstyle t=1$} 
        node[midway,below] {$\scriptstyle t:=0$};
      \draw[latex'-] (a) -- +(-1,0) node[midway,above] {$\scriptstyle t=0$};
      \end{tikzpicture}
    \caption{Encoding incrementation with a rational
      parameter}\label{fig-incr-robust} 
  \end{figure}
  This way, exactly one time unit elapses in this module, and
  clock~$x_i$~is decreased by~$p$, which corresponds to
  incrementing~$c_i$. In \cref{fig-incr-robust}, the upper branch
  corresponds to the case when $c_1+1\leq c_2$ and the lower branch to
  the case when $c_2\leq c_1+1$. When both values are equal, both
  branches can be taken, with the same effect. Decrementing is handled
  similarly. Finally, notice that the use of the constraint~$x_i=1+p$
  can be easily avoided, at the expense of an extra clock.

  One~easily proves that if a (deterministic) two-counter
  machine~$\calM$ halts, then by
  writing~$P$ for the maximal counter value reached during its finite
  computation, the PTA above has a path to the halting location as soon as
  $0<p\leq 1/P$. Conversely, assume that the machine does not halt, and fix a
  parameter value~$0<p<1$. If some counter of the machine eventually
  exceeds~$1/p$, then at that moment in the corresponding execution in the
  associated PTA, the value of~$t$ when $x_i=1+p$ will be larger than~$1$, and
  the automaton will be in a deadlock. If~the counters remain bounded
  below~$1/p$, then the execution of the two-counter machine will be simulated
  correctly, and the halting state will not be reached.

  We~now adapt this construction to our language preservation problem. We~have
  to forbid the infinite non-halting run mentioned above. For this, we~add a
  third counter, which will be incremented every other step of the resulting
  three-counter machine, in the very same way as in the proof of
  \cref{theorem:trace-preservation-undec}.
  We~then have the property
  that if~$\calM$ does not halt, the~simulation in the associated PTA will be finite, for all non-zero parameter valuation.
  Adding locations~$s_\init$ and~$s_{\infty}$ as in
  \cref{figure:language-preservation}, we~get the result that the
  two-counter machine~$\calM$ halts if, and only~if, there is a parameter
  value~$v_0(p)>0$ such that all values $v(p)$ between~$0$ and~$v_0(p)$ give
  rise to timed automata $v(\calA)$ accepting the same language as for $v_0(p)=0$.

  Finally, we~notice that this reduction works even if we impose a positive
  upper bound on~$p$ (typically~1). 
\end{proof}

\subsection{Undecidability of the Robust Trace Preservation Problem}



Combining \cref{theorem:robust-language-preservation:undecidable:bPTA}
and the arguments of \cref{sec-trace}, we~get:
\begin{restatable}{theorem}{thmrobtr}
\label{theorem:robust-trace-preservation:undecidable:bPTA}
The robust trace-preservation problem is undecidable for PTA with one
(possibly bounded) parameter.
\end{restatable}
Both proofs developed in \cref{sec-trace} can be applied here: 
\begin{itemize}
\item the first proof, using diagonal constraints (\cref{ss:one-location}), applies as the PTA built
  above does not contain zero-delay cycles;
\item the second proof (\cref{ss:bounded-locations}) also applies, by using one clock~$s_i$ per
  location~$\ell_i$ of the two-counter machine with the encoding that the
  clock corresponding to the current location~$\ell_i$ is the only clock~$s_i$
  with value less than or equal to~$1$.
	Notice that we keep a bounded number of locations in that case. 
\end{itemize}


\subsection{Undecidability over Bounded Time}
\label{ss:time-bounded}


Let us now consider decision problems over bounded time, \ie{} when 
the property must additionally be satisfied within $T$ time units, for a given constant $T\geq 0$, and the system can thus be studied only inside that time frame.
We first prove that the EF-emptiness problem is undecidable for PTA with three clocks and two rational-valued parameters, over bounded (dense) time.
This result was already mentioned in~\cite{Jovanovic13}; however, we provide here a full (and different) proof.
Most importantly, this result will then be used to prove the undecidability of the problems considered earlier in this section in the time-bounded setting too.

\begin{theorem}\label{theorem:robust-language-preservation:undecidable:bPTA:time-bounded}
	The EF-emptiness problem is undecidable over bounded time for
        PTA with three clocks and two parameters.
\end{theorem}
\begin{proof}
	We reduce from the halting of a two-counter machine.
	Let~us reuse the encoding of the proof of \cref{theorem:robust-language-preservation:undecidable:bPTA}, and modify it as follows:
	\begin{itemize}
		\item The system is studied over 1 time unit (\ie{} $T=1$);
		\item We rename $\param$ into~$\param_2$;
		\item We replace any occurrence of ``1'' with a new parameter~$\param_1$; intuitively, this parameter will be small enough (compared to~1) to encode the length of the execution of the machine; in addition, $\param_1$ must be sufficiently large when compared to~$\param_2$, so that $\param_2$ can encode the maximum value of the counters.
		With our variables replacing, when $t = 0$, we now have the encoding $x_i = p_1 - p_2 \cdot c_i$.
		For any positive valuation of $p_2$, the maximum value of the counter that our encoding can support therefore becomes $p_1 / p_2$.\ea{added to answer reviewer 1}\label{newtext:encoding:TB}
	\end{itemize}
	We give the modified increment gadget in \cref{figure:time-bounded:increment} and the decrement gadget in \cref{figure:time-bounded:decrement}.
	The increment gadget requires $\param_1$ time units to be traversed, and the decrement gadget requires $\param_1 + \param_2$ time units.
	The zero-test gadget (which requires $0$ time unit in \cref{theorem:robust-language-preservation:undecidable:bPTA}) is modified in an appropriate manner to require $\param_1$ time units (see \cref{figure:time-bounded:0-test}).
	Now, since any gadget requires at least $\param_1$ time units, it is clear that, for any value of $\param_1 > 0$, the number of operations that the machine can perform is finite, since the system executes over $1$~time unit.
	
	The initial gadget constrains $\param_1$ to be strictly positive, and ensures that $x_1 = x_2 = 1$ while $t = 0$.
    In the gadget for incrementation, the upper branch corresponds to
    the case when $c_1 + 1 \leq c_2$ and the lower branch to $c_2 \leq
    c_1 + 1$. Similarly, in the decrementation gadget, the upper
    branch corresponds to $c_1 \leq c_2 + 1$ and the lower branch to
    $c_2 + 1\leq c_1$. Finally, in the zero-test gadget, the~upper
    branch corresponds to $c_1\leq c_2$ and the lower branch to
    $c_2\leq c_1$. In~all these cases when both values are equal,
    both branches can be taken, with the same effect.

	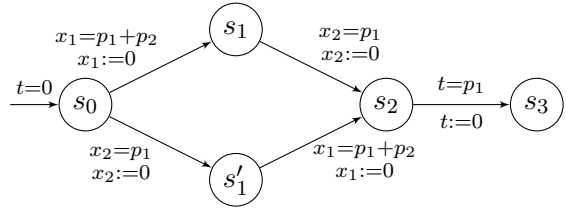
\begin{figure} [t]
	\centering
	\begin{tikzpicture}
		\draw (0,0) node[rond] (a) {} node {$s_0$};
		\draw (2,1) node[rond] (b) {} node {$s_1$};
		\draw (2,-1) node[rond] (c) {} node {$s'_1$};
		\draw (4,0) node[rond] (d) {} node {$s_2$};
		\draw (6,0) node[rond] (e) {} node {$s_3$};
		\draw[-latex'] (a) -- (b) node[midway,above left=-6pt]
			{$\genfrac{}{}{0pt}{1}{x_1=\param_1+\param_2
            }{x_1:=0}$};
		\draw[-latex'] (a) -- (c) node[midway,below left=-3pt]
			{$\genfrac{}{}{0pt}{1}{x_2=\param_1}{x_2:=0}$};
		\draw[-latex'] (c) -- (d) node[midway,below right=-6pt]
			{$\genfrac{}{}{0pt}{1}{x_1=\param_1+\param_2
            }{x_1:=0}$};
		\draw[-latex'] (b) -- (d) node[midway,above right=-3pt] 
			{$\genfrac{}{}{0pt}{1}{x_2=\param_1}{x_2:=0}$};
		\draw[-latex'] (d) -- (e) node[midway,above] {$\scriptstyle t=\param_1$} 
			node[midway,below] {$\scriptstyle t:=0$};
		\draw[latex'-] (a) -- +(-1,0) node[midway,above] {$\scriptstyle t=0$};
	\end{tikzpicture}

	\caption{Encoding incrementation with rational parameters over bounded time}
	\label{figure:time-bounded:increment}
	\end{figure}
	
	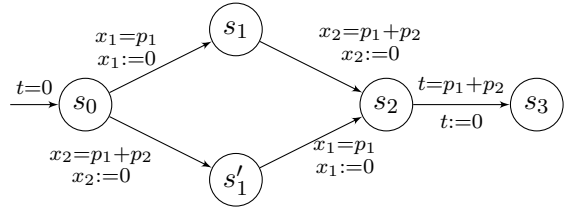
\begin{figure} [t]
	\centering
	\begin{tikzpicture}
		\draw (0,0) node[rond] (a) {} node {$s_0$};
		\draw (2,1) node[rond] (b) {} node {$s_1$};
		\draw (2,-1) node[rond] (c) {} node {$s'_1$};
		\draw (4,0) node[rond] (d) {} node {$s_2$};
		\draw (6,0) node[rond] (e) {} node {$s_3$};
		\draw[-latex'] (a) -- (b) node[midway,above left=-6pt]
			{$\genfrac{}{}{0pt}{1}{x_1=\param_1
            }{x_1:=0}$};
		\draw[-latex'] (a) -- (c) node[midway,below left=-3pt]
			{$\genfrac{}{}{0pt}{1}{x_2=\param_1+\param_2}{x_2:=0}$};
		\draw[-latex'] (c) -- (d) node[midway,below right=-6pt]
			{$\genfrac{}{}{0pt}{1}{x_1=\param_1
            }{x_1:=0}$};
		\draw[-latex'] (b) -- (d) node[midway,above right=-3pt] 
			{$\genfrac{}{}{0pt}{1}{x_2=\param_1+\param_2}{x_2:=0}$};
		\draw[-latex'] (d) -- (e) node[midway,above] {$\scriptstyle t = \param_1 + \param_2$} 
			node[midway,below] {$\scriptstyle t:=0$};
		\draw[latex'-] (a) -- +(-1,0) node[midway,above] {$\scriptstyle t=0$};
	\end{tikzpicture}

	\caption{Encoding decrementation with rational parameters over bounded time}
	\label{figure:time-bounded:decrement}
	\end{figure}
	
	\begin{figure}[tb]
		\centering
		\begin{tikzpicture}
		\draw (0,0) node[rond] (a) {} node {$s_0$};
		\draw (2,1) node[rond] (b) {} node {$s_1$};
		\draw (2,-1) node[rond] (c) {} node {$s'_1$};
		\draw (4,0) node[rond] (d) {} node {$s_2$};
		\draw (6,0) node[rond] (e) {} node {$s_3$};
		\draw[-latex'] (a) -- (b) node[midway,above left=-6pt]
			{$\genfrac{}{}{0pt}{1}{x_1=p_1}{x_1:=0}$};
		\draw[-latex'] (a) -- (c) node[midway,below left=-3pt]
			{$\genfrac{}{}{0pt}{1}{x_2=p_1}{x_2:=0}$};
		\draw[-latex'] (c) -- (d) node[midway,below right=-6pt]
			{$\genfrac{}{}{0pt}{1}{x_1=p_1}{x_1:=0}$};
		\draw[-latex'] (b) -- (d) node[midway,above right=-3pt] 
			{$\genfrac{}{}{0pt}{1}{x_2=p_1}{x_2:=0}$};
		\draw[-latex'] (d) -- (e) node[midway,above] {$\genfrac{}{}{0pt}{1}{t=p_1}{\et x_1=p_1}$}
			node[midway,below] {$\scriptstyle t:=0$};
		\draw[latex'-] (a) -- +(-1,0) node[midway,above] {$\scriptstyle t=0$};
		\end{tikzpicture}
		\caption{Encoding 0-test over bounded-time}\label{figure:time-bounded:0-test} 
	\end{figure}
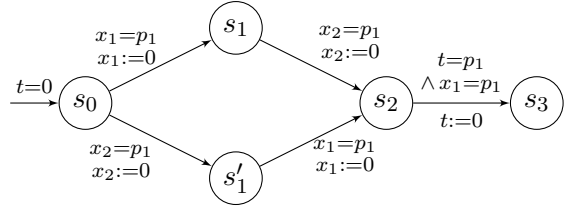

  Let us prove that there exists a run reaching \shalt{} in at most 1 time unit if, and only~if, the two-counter machine halts.
	\begin{enumerate}
		\item Assume that the machine does not halt.
            If $\param_1\leq 0$ then the initial gadget cannot be traversed. Now consider $\param_1>0$.
			In this case, whatever the value of the parameters, after a maximum number of steps (at most $\frac{1}{\param_1}$), one full time unit will elapse without the system reaching \shalt{}.
			In addition, if the value of~$\param_2$ is not small enough to encode the maximum value of the counters over these 
            steps, an increment gadget will block, again without reaching \shalt{}.
		Hence if the two-counter machine does not halt, \shalt{} cannot be reached within 1 time unit.
		
	      \item Assume that the machine halts: in this case, if
                $c$ is the maximum value of both $C_1$ and $C_2$ over
                the (necessarily finite) halting execution of the
                machine, and if $\runlength$ is the length of this
                execution, and if $c>0$, then for valuations such that
                $\param_2 \leq \frac{\param_1}{c}$ and sufficiently
                small valuations of $\param_1$ and $\param_2$ (at most
                $\param_1 + \param_2 \leq \frac{1}{\runlength}$ as any
                gadget takes at most $\param_1+\param_2$ time units),
                then there exists one run that correctly simulates the
                machine, and eventually reaches \shalt{}.
		This set of valuations is non-empty: for example
		if $c>0$, then we can choose
                $\param_2 = \frac{1}{(c+1)\runlength}$
                and $\param_1 = \frac{c}{(c + 1) \runlength}$
		(since in the worst case, the sequence duration
                is $\runlength(\param_1 + \param_2)$);
                if $c=0$ then $\param_1 = \param_2 = \frac{1}{\runlength}$ (to allow for up to $\frac{\runlength}{2}$ decrements or $\runlength$ zero-tests).
	Hence, if the two-counter machine halts, there exist parameter valuations for which a run reaches \shalt{} within $1$ time unit.\qedhere
	\end{enumerate}
\end{proof}


\begin{theorem}
	The robust language-preservation problem for PTA with two (possibly bounded) parameters is undecidable over bounded time.
\end{theorem}
\begin{proof}
	By reusing the encoding of \cref{theorem:robust-language-preservation:undecidable:bPTA:time-bounded} in the reasoning of the proof of \cref{theorem:robust-language-preservation:undecidable:bPTA}.
\end{proof}

Following a similar reasoning, we can also show the undecidability of all the other problems considered in this section in the time-bounded setting.

\begin{remark}\label{remark:timebounded-HA}
	Our undecidability results can be put into perspective with the decidability results for the larger class of hybrid automata of~\cite{BDGORW13}.
	In~\cite{BDGORW13}, time-bounded reachability is proved decidable for a subclass of hybrid automata with monotonic (either non-negative or non-positive) rates: parametric timed automata can fit into this framework: clocks and parameters all have non-negative rates (1~for clocks, and 0~for parameters), with the exception of the initialization phase: in~that phase, we~let time elapse until the parameters (growing at rate~1) reach their value,
        and then set their rates to~0;
        we~then reset all clocks and start the \emph{real} execution of the automaton.
	However, to compare clocks and parameters together in a hybrid automaton, one needs diagonal constraints---which~are not allowed in~\cite{BDGORW13}.
	As we showed that our undecidability results (notably \cref{theorem:robust-language-preservation:undecidable:bPTA:time-bounded}) hold over bounded-time with only two parameters, one can revisit the result of~\cite{BDGORW13} as follows: allowing only two variables (our parameters) in diagonal constraints, with only two locations with a non-zero rate (the~initialization locations for these parameters) makes the decidable problem of~\cite{BDGORW13} undecidable.
\end{remark}

\section{A Semi-Algorithm for the Trace Preservation Synthesis}\label{section:algo}

In this section, we propose a semi-algorithm that solves the following
\emph{parameter-synthesis} problem: ``\emph{given a PTA~$\A$ and a parameter
  valuation~$\pval$, synthesize parameter valuations that yield the same
  language (or~trace~set) as~$\pval$}''.

The inverse method proposed in~\cite{ACEF09} outputs a parameter constraint
that is a correct but incomplete answer to the trace-preservation 
problem. Below, we rewrite this algorithm so that, whenever it terminates,
it~outputs a correct answer for any~PTA, and a complete answer for
deterministic PTA.

\subsection{The Algorithm \TPS}

We give $\TPS(\A, \pval)$ in \cref{algo:IM}. $\TPS$~maintains two
constraints: $\Kgood$~is the intersection of the parameter constraints
associated with the $\pval$-compatible symbolic states met, whereas
$\Kbad$ is the union\footnote{This union of constraints can be seen
(and implemented) as a finite list of convex constraints.} of the
parameter constraints associated with all $\pval$-incompatible
symbolic states. $\TPS$ also maintains two sets of symbolic states, \viz{} the
set $\States$ of all symbolic states met, and the set $\Snew$ of
symbolic states met at
the latest iteration of the \textbf{while} loop. $\TPS$~is a
breadth-first search algorithm, that iteratively explores the symbolic
state space. Whenever a new symbolic state is met, its $\pval$-compatibility is
checked~(\cref{algo:IM:check}). If it is $\pval$-compatible, its
projection onto the parameters is added
to~$\Kgood$~(\cref{algo:IM:Kgood}). Otherwise, its~projection onto the
parameters is added to~$\Kbad$~(\cref{algo:IM:Kbad}), and the symbolic state is
discarded from~$\Snew$~(\cref{algo:IM:discard}), \ie{} its successors
will not be explored. When no new symbolic states can be explored, \ie{} the
set $\Snew$ is either empty or contains only symbolic states explored
earlier~(\cref{algo:IM:fixpoint}), the intersection of
$\pval$-compatible parametric constraints and the negation of the
$\pval$-incompatible parametric constraints is returned
(\cref{algo:IM:return}). Otherwise, the algorithm explores one step
further in depth~(\cref{algo:IM:i++}).
$\TPS$ is implemented in the \imitator{} software~\cite{AFKS12}.

\begin{algorithm}[t]
	\SetKwInOut{Input}{input}
	\SetKwInOut{Output}{output}

	\Input{PTA~$\A$, parameter valuation $\pval$}
	\Output{Constraint $\K$ over the parameters}

	\BlankLine
	$
		\Kgood \algoAssign \KTrue \,; \ \ 
		\Kbad \algoAssign \KFalse \,; \ \ 
		\Snew \algoAssign \{ \sinit \} \,; \ \ 
		\States \algoAssign \varnothing
	$

	\While{$\BTrue$}{
		\ForEach{symbolic state $(\loc, \C) \in \Snew$}
		{
			\lIf{$\pval \models \projectP{\C}$\nllabel{algo:IM:check}}{
                          \nllabel{algo:IM:Kgood}$\Kgood \algoAssign \Kgood \land \projectP{\C}$
			}

			\quad \lElse{
				$\Kbad \algoAssign \Kbad \lor \projectP{\C}$\nllabel{algo:IM:Kbad}
				;\ \
				$\Snew \algoAssign \Snew \setminus \{ (\loc, \C) \} $\nllabel{algo:IM:discard}
			}

		} 
		
		\lIf{$\Snew \subseteq \States $ \nllabel{algo:IM:fixpoint} } 
			{\Return{$\Kgood \land \neg \Kbad $} \nllabel{algo:IM:return} }

		$
			\States \algoAssign \States \cup \Snew \, ; \ \ 
			\Snew \algoAssign \Succ(\Snew)
		$ \nllabel{algo:IM:i++}

	} 

	\caption{$\TPS(\A, \pval)$}
	\label{algo:IM}
\end{algorithm}
%

\subsection{Soundness of \TPS}

\cref{theorem:TPS:correctness} states that, in~case $\TPS(\A, \pval)$
terminates, its result is correct.

\begin{restatable}
  {theorem}{TheoremCorrectnessTPS}
\label{theorem:TPS:correctness}
    Let $\A$ be a PTA, let $\pval$ be a parameter valuation.
	Assume $\TPS(\A, \pval)$ terminates with constraint $\K$.
	Then
$\pval \models \K$, and
for all $\pval' \models \K$, $\Traces(\valuate{\A}{\pval'}) = \Traces(\valuate{\A}{\pval})$.
\end{restatable}

Let us prove \cref{theorem:TPS:correctness} in the following.
We first recall below a useful result stating that the projection onto the parameters of a constraint can only become more strict along a run.

\begin{lemma}\label{lemma:k-stricter}
	Let $\A$ be a PTA, and let $\run$ be a run of~$\A$ reaching $(\loc, \C)$.
	Then, for any successor $(\loc', \C')$ of $(\loc, \C)$, we have $\projectP{\C'} \subseteq \projectP{\C}$.
\end{lemma}
\begin{proof}
	Let $(\loc, \C)$ be a symbolic state.
	Let $e=(\loc,\guard,\action,\resets,\loc')$ be a transition.
	Let $(\loc', \C')$ be the successor of~$(\loc, \C)$ via~$e$.
	Recall that, from the definition of the symbolic semantics, $\C' = \timelapse{\big(\reset{(\C \land \guard)}{\resets}\big )} \cap \invariant(\loc')$.
	
	Let $\pval \models \projectP{\C'}$.
	Then there exists $u'\models C'$ and a clock valuation $\clockval'$ such that $u'_{|\Param}=\pval$ and $u'_{|\Clock}=\clockval'$.
	From the definition of the $\Succ$ operator, there exists a clock valuation $\clockval$, $u \models \C$ and $d \geq 0$ such that $\clockval' = \clockval[R] + d$, $u_{|\Param} = \pval$ and $u_{|\Clock}= \clockval$.
	So $\pval \models \projectP{\C}$.
%
\end{proof}
\ea{légèrement hand-waiving, ici. Si vous avez des arguments précis sur les première et dernière implications, n'hésitez pas.}
\dl{et pourquoi pas simplement (à noter que comme ce sont des contraintes partout l'inclusion est plutôt une implication non?): Let $v\models \projectP{C'}$. Then there exists $u'\models C'$ and a clock valuation $w'$ such that $u'_{|\Param}=v$ and $u'_{|\Clock}=w'$. From the definiton of the $\Succ$ operator, there exists a clock valuation $w$, $u\models C$ and $d\geq 0$ such that $w'=w[R] + d$, $u_{|\Param}=v$ and $u_{|\Clock}=w$. So $v\models \projectP{C}$.}\ea{le souci est que le $\Succ$ est le $\Succ$ symbolique, et qu'on n'a pas (ici) lié les symboliques aux concrets ; par contre, avec l'aide de \cref{prop:run-equivalence,prop:run-nonparam-param} ci-dessous, y'a ptêt moyen. On prend une valeur dans la contrainte, on montre qu'il y a une valeur concrète, on remonte en arrière avec le succ concret, et on dit que du coup cette valeur existe dans un run équivalent. Mais le problème est que ça pourrait être un \emph{autre} run symbolique équivalent (puisque cette notion n'a pas été très bien définie, ou alors faudrait montrer l'unicité…) Qu'en penses-tu ?}
\dl{j'ai l'impression de n'avoir vraiment utilisé que la def des opérateurs future, reset, and co tels que présentés p. 4, non ?}
\ea{OK}

We now recall below two results from \cite{HRSV02}.

\begin{proposition}\label{prop:run-equivalence}
	Let $\A$ be a PTA, and let $\run$ be a run of~$\A$ reaching
	$(\loc, \C)$.  Let $\pval$ be a parameter valuation.  There
	exists an equivalent run in~$\valuate{\A}{\pval}$ if, and
	only~if, $\pval \models \projectP{\C}$.
\end{proposition}
\begin{proof}
    From \cite[Propositions 3.17 and 3.18]{HRSV02}.
\end{proof}

\begin{proposition}\label{prop:run-nonparam-param}
	Let $\A$ be a PTA, let $\pval$ be a parameter valuation. 
	Let $\run$ be a run of~$\valuate{\A}{\pval}$ reaching $(\loc, \clockval)$.
	Then there exists an equivalent symbolic run in~$\A$ reaching $(\loc, \C)$, with $\pval \models \projectP{\C}$.
\end{proposition}
\begin{proof}
	From \cite[Proposition 3.18]{HRSV02}.
\end{proof}

Before proving  \cref{theorem:TPS:correctness},
we need some intermediate results.

\begin{lemma}\label{lemma:TPS:pval-models-K}
	Let $\A$ be a PTA, let $\pval$ be a parameter valuation.
	Assume $\TPS(\A, \pval)$ terminates with constraint $\K$.
	Then $\pval \models \K$.
\end{lemma}
\begin{proof}
	By construction, all constraints added to $\Kgood$ are $\pval$-compatible, hence their intersection is $\pval$-compatible.
	By construction, all constraints added to $\Kbad$ are $\pval$-incompatible, hence their union is $\pval$-incompatible; hence the negation of their union is $\pval$-compatible.
	This gives that $\pval \models \Kgood \land \neg \Kbad$, thus $\pval \models \K$.
\end{proof}

\begin{lemma}\label{lemma:TPS:correctness}
	Let $\A$ be a PTA, let $\pval$ be a parameter valuation.
	Assume $\TPS(\A, \pval)$ terminates with constraint $\K$.
	Then for all $\pval' \models \K$, we~have
        $\Traces(\valuate{\A}{\pval'}) = \Traces(\valuate{\A}{\pval})$.
\end{lemma}
\begin{proof}
	Let $\pval' \models \K$.
	\begin{itemize}[align=left]
	 \item[($\subseteq$)] Let $\run'$ be a run of $\valuate{\A}{\pval'}$, reaching a symbolic state $(\loc, \clockval')$.
		From \cref{prop:run-nonparam-param}, there exists an equivalent run in~$\A$ reaching a symbolic state $(\loc, \C')$, with $\pval' \models \projectP{\C'}$.
		
		We will now prove by \emph{reductio ad absurdum} that $\pval \models \projectP{\C'}$.
		Assume 
		$\pval \not\models \projectP{\C'}$.
		Hence $(\loc, \C')$ is either a $\pval$-incompatible symbolic state met in $\TPS(\A, \pval)$, or the successor of some $\pval$-incompatible symbolic state met in $\TPS(\A, \pval)$.
		\begin{enumerate}
			\item Assume $(\loc, \C')$ is a $\pval$-incompatible symbolic state met in $\TPS(\A, \pval)$.
				By construction, $\projectP{\C'}$ has been added to $\Kbad$ (\cref{algo:IM:Kbad} in \cref{algo:IM}),
					hence $\projectP{\C'} \subseteq \Kbad $
					hence $\neg \Kbad \cap \projectP{\C'} = \emptyset$
					hence $(\Kgood \land \neg \Kbad) \cap \projectP{\C'} = \emptyset$
					hence $\K \cap \projectP{\C'} = \emptyset$.
					This contradicts that $\pval' \models \K$.
			\item Assume $(\loc, \C')$ is a $\pval$-incompatible symbolic state not met in $\TPS(\A, \pval)$, \ie{} it belongs to some path starting from a $\pval$-incompatible symbolic state $(\loc'', \C'')$ met in $\TPS(\A, \pval)$.
			From \cref{lemma:k-stricter}, $\projectP{\C'} \subseteq \projectP{\C''}$, and hence $\projectP{\C'} \subseteq \projectP{\C''} \subseteq \Kbad $; then we apply the same reasoning as above to prove that $\K \cap \projectP{\C'} = \emptyset$, which contradicts that $\pval' \models \K$.
		\end{enumerate}
		Hence $\pval \models \projectP{\C'}$.
		Now, from \cref{prop:run-equivalence}, there exists an equivalent run in $\valuate{\A}{\pval}$, which gives that $\Traces(\valuate{\A}{\pval'}) \subseteq \Traces(\valuate{\A}{\pval})$.
		
	 \item[($\supseteq$)] Let $\run$ be a run of $\valuate{\A}{\pval}$, reaching a symbolic state $(\loc, \clockval)$.
		From \cref{prop:run-nonparam-param}, there exists an equivalent run in~$\A$ reaching a symbolic state $(\loc, \C)$, with $\pval \models \projectP{\C}$.
		From the fixpoint condition of \cref{algo:IM}, all $\pval$-compatible symbolic states of~$\A$ have been explored in~$\TPS(\A, \pval)$, hence $(\loc, \C) \in \States$, where $\States$ is the set of symbolic states explored just before termination of~$\TPS(\A, \pval)$.
		By construction, $\K \subseteq \projectP{\C}$; since $\pval' \models \K$ then $\pval' \models \projectP{\C}$.
		Hence, from \cref{prop:run-equivalence}, there exists an equivalent run in $\valuate{\A}{\pval'}$, which gives that $\Traces(\valuate{\A}{\pval'}) \supseteq \Traces(\valuate{\A}{\pval})$.
    \qedhere
	\end{itemize}
\end{proof}
\noindent
\Cref{theorem:TPS:correctness} immediately follows
from \cref{lemma:TPS:pval-models-K,lemma:TPS:correctness}.

%
%

\subsection{Completeness of \TPS}

We now state the completeness of $\TPS$ for \emph{deterministic}~PTA.

\begin{restatable}[completeness of $\TPS$]{theorem}{TheoremCompletenessTPS}
\label{theorem:TPS:completeness}
    Let $\A$ be a deterministic PTA, let $\pval$ be a parameter valuation.
	Assume $\TPS(\A, \pval)$ terminates with constraint $\K$.
	Then $\pval' \models \K$ if, and only~if,
        $\Traces(\valuate{\A}{\pval'}) = \Traces(\valuate{\A}{\pval})$.
\end{restatable}


    \begin{proof}
    \Cref{theorem:TPS:correctness} entails that
    $\Traces(\valuate{\A}{\pval'}) = \Traces(\valuate{\A}{\pval})$
    whenever $\pval' \models \K$. We~prove the other implication.
  Let $\pval'$ be a parameter valuation such that $\Traces(\valuate{\A}{\pval'}) = \Traces(\valuate{\A}{\pval})$.
			The result comes from the fact that, in a deterministic (P)TA, the equality of trace sets implies the equivalence of runs.
			Hence we can prove a stronger result, that is $\TPS(\A, \pval) = \TPS(\A, \pval')$.
			Indeed, $\TPS(\A, \pval')$ proceeds by exploring symbolic states similarly to $\TPS(\A, \pval)$.
            From \cref{prop:run-equivalence},
            the $\pval$-incompatible and $\pval$-compatible symbolic states met in $\TPS(\A, \pval')$ will be the same as in $\TPS(\A, \pval)$, and hence the constraints $\Kgood$ and $\Kbad$ will be the same too.
			Hence $\TPS(\A, \pval) = \TPS(\A, \pval')$, which trivially gives that $\pval' \models \TPS(\A, \pval)$.
	\end{proof}



%
%
\begin{remark}\label{remark:incompleteness}
	The incompleteness of $\TPS$ for nondeterministic PTA is easily seen:
	Consider the PTA~$\A$ in \cref{figure:example:nondet}.
	Clearly, the upper transition from~$\loc_0$ to~$\loc_1$ can only be taken if $\param \leq 1$, and the lower transition if $\param > 1$.
	Consider the valuation $\pval$ assigning 0 to~$\param$.
	The~(unique) trace in $\valuate{\A}{\pval}$ is $(\loc_0, a, \loc_1)$.
	
	Running $\TPS(\A, \pval)$, we get two symbolic states corresponding to~$\loc_1$:
	\begin{itemize}
		\item From the upper transition, we get $(\loc_1, x \geq 1 \land p \leq 1)$: this symbolic state is $\pval$-compatible; $\Kgood$ is thus updated to the projection of this symbolic state onto~$\Param$, \ie{} $\Kgood = p \leq 1$.
		\item From the lower transition, we get $(\loc_1, x \geq 1 \land p > 1)$: this symbolic state is $\pval$-incompatible, and therefore $\Kbad$ is updated to $p > 1$.
	\end{itemize}
	Eventually, $\Kgood \land \neg \Kbad$ is returned, that is $p \leq 1 \land \neg (p > 1)$ which gives $p \leq 1$.
	However, the trace $(\loc_0, a, \loc_1)$ is in fact possible for \emph{any} parameter valuation $p \geq 0$, and therefore the result output by $\TPS(\A, \pval)$ is not complete.
\end{remark}

\begin{figure}[tb]
	\centering
	\small
		\begin{tikzpicture}[->, >=stealth', auto, thin,scale=1.4]

			\node[rond] (l0) at(0,0) {$\loc_0$};
			\node[rond] (l1) at (3,0) {$\loc_1$};
			
			\path
				(l0) edge[bend left] node[above] {\begin{tabular}{c}
						$\clock = 1 \land \clock \geq \param$\\
						$a$ \\
					\end{tabular}} (l1)
				(l0) edge[bend right] node[below] {\begin{tabular}{c}
						$\clock = 1 \land \clock < \param$\\
						$a$ \\
					\end{tabular}} (l1)
				;
		\end{tikzpicture}
	\caption{Non-deterministic PTA for which $\TPS$ is not complete}
	\label{figure:example:nondet}
\end{figure}
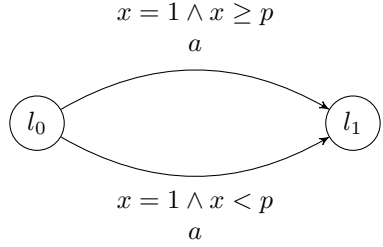

%
%



\section{Decidability Results for Subclasses of PTA}\label{section:particular}

In this section, we first prove the finiteness of the parametric zone graph of 1-clock
PTA over both discrete and rational time (\cref{ss:1clock}). We then study the
(un)decidability of the language and trace preservation emptiness problems for
deterministic 1-clock PTA (\cref{ss:dec-1clock}), L/U-PTA (\cref{ss:L/U}) and
deterministic 1-parameter L-PTA and U-PTA (\cref{ss:L-U-PTA}).

\subsection{1-Clock PTA}\label{ss:1clock}

In this section, we restrict the number of clocks of a PTA, without any restriction on the number of parameters.
In fact, we even slightly extend the definition of PTA, by allowing parametric linear terms in guards and invariants.

\begin{definition}[1-clock PTA]
\label{definition:1clock-PTA}
	An \emph{extended 1-clock PTA} (1cPTA for~short) is a PTA with only
	one clock and possibly several parameters, and allowing guards and
	invariants of the form $\clock \compOp \sum_i \alpha_i \param_i + c$,
	with $\param_i \in \Param$ and $\alpha_i, c \in \grandz$.
\end{definition}

We show below that the parametric zone graph for 1cPTA is finite.
In~\cite{AHV93}, it is shown that the set of parameters for which there exists a run that can reach a given location can be computed for PTA over discrete time with only one parametric clock and arbitrarily many non-parametric clocks.
Here, we lift the assumption of discrete time, we allow more general guards and invariants, and the finiteness of the parametric zone graph allows to synthesize valuations for more complex properties than pure reachability; however, we only consider a single (parametric) clock.
Adding non-parametric clocks in this setting (perhaps reusing a construction used in~\cite{BBLS15}) is the subject of future work.

\begin{definition}
\label{definition:1clock-C}
	Given a 1cPTA~$\A$, a \emph{1-clock symbolic constraint} is a constraint over~$\Clock \cup \Param$ of the form
          $\bigwedge_{i} (\lt_i \compOp \clock) \land \bigwedge_{j} (\lt^1_j \compOp \lt^2_j)$,
	where $i,j \in \grandn$, $\clock$ is the unique clock of~$\A$, and $\lt_i, \lt^1_j, \lt^2_j$ are parametric linear terms (\ie{} of the form $\sum_i \alpha_i \param_i + c$) either appearing in guards and invariants of~$\A$, or equal to~0, and such that $\lt^1_j, \lt^2_j$ are all different from each other.
	We~denote by $\oneClockSC(\A)$ the set of 1-clock symbolic constraints of~$\A$.
\end{definition}


\begin{restatable}{lemma}{LemmaCSC}
\label{lemma:shape-constraints}
	Let $\A$ be a 1cPTA.
	Let $(\loc, \C)$ be a reachable symbolic state of~$\A$.
	Then $\C \in \oneClockSC(\A)$.
\end{restatable}


\begin{proof}
	By induction on the length of the runs.
	\begin{description}
		\item[Base case] A run of length 0 consists of the sole initial symbolic state.
			According to the semantics of PTA, this symbolic state is $(\locinit, \Cinit)$, where $\Cinit$ is $\timelapse{(\Clock = 0)} \land \invariant(\locinit)$, \ie{} $\clock \geq 0 \land \invariant(\locinit)$.
			From \cref{definition:1clock-PTA}, $\invariant(\locinit)$ is of the form $\bigwedge_{i}\lt_i \compOp \clock$, with $\lt_i$ parametric linear terms of~$\A$, hence $\invariant(\locinit) \in \oneClockSC(\A)$.
			Furthermore, $\clock \geq 0$ obviously belongs to $\oneClockSC(\A)$.
			Hence the initial constraint belongs to $\oneClockSC(\A)$.
			
		\item[Induction step]
			Consider a run of length $n$ reaching the symbolic state $(\loc, \C)$, and assume $\C$ is of the form
			$$\bigwedge_{i}(\lt_i \compOp \clock) \land \bigwedge_{j} (\lt^1_j \compOp \lt^2_j)\text{.}$$
			Let $(\loc', \C')$ be a successor of $(\loc, \C)$ through the $\Succ$ operation, for some edge $(\loc,\guard,\action,\resets,\loc')$.
			Recall that $\C' = \timelapse{\big(\reset{(\C \land \guard)}{\resets}\big )} \cap \invariant(\loc')$.
			Let us the consider the different operations sequentially.
			
			\begin{description}
				\item[Guard]
					From \cref{definition:1clock-PTA}, a guard is of the form $\clock \compOp \sum_i \alpha_i \param_i + c$, with $\param_i \in \Param$ and $\alpha_i \in \grandz$; hence $\guard \in \oneClockSC(\A)$.
					Since $\C  \in \oneClockSC(\A)$ by induction hypothesis, then $\C \land \guard \in \oneClockSC(\A)$.
				
				\item [Reset] Then, $\reset{(\C \land \guard)}{\resets}$ is equivalent to removing $\clock$ in $\C \land \guard$ (using variable elimination technique such as Fourier-Motzkin) and adding a fresh equality $\clock = 0$.
				The elimination of~$\clock$ will leave the set of parametric inequalities (\ie{} $\bigwedge_{j} \lt^1_j \compOp \lt^2_j$) unchanged.
                As for the inequalities containing $\clock$ (\ie{} $\bigwedge_{i}\lt_i \compOp \clock$), the elimination of $\clock$ will lead to the disappearance of some of the $\lt_i$, as well as the creation of new inequalities of the form $\lt_i \compOp \lt_{i'}$, which will be added to the set of parametric inequalities (see, \eg{}~\cite{Schrijver86}).
				Finally, adding $\clock = 0$ (which belongs to  $\oneClockSC(\A)$) makes $\reset{(\C \land \guard)}{\resets}$ remain in $\oneClockSC(\A)$.
				
				\item[Time elapsing] The time elapsing will remove some upper bounds on~$\clock$, which leads to the disappearance of some of the inequalities, and hence makes $\timelapse{\big(\reset{(\C \land \guard)}{\resets}\big )}$ still belong to $\oneClockSC(\A)$.
				
				\item[Addition of the target invariant] The target invariant $\invariant(\loc')$ adds new inequalities, all belonging to $\oneClockSC(\A)$, hence $\timelapse{\big(\reset{(\C \land \guard)}{\resets}\big )} \cap \invariant(\loc') \in \oneClockSC(\A)$.
			\end{description}
			Hence, $\C' \in \oneClockSC(\A)$.
      \qedhere
	\end{description}
\end{proof}

\begin{theorem}
\label{theorem:finiteness-zone-graph}
	The parametric zone graph of a 1cPTA is finite.
\end{theorem}
\begin{proof}
	From \cref{lemma:shape-constraints}, each symbolic state of a 1cPTA~$\A$ belongs to $\oneClockSC(\A)$.
	Due to the finite number of linear terms in the guards and invariants in~$\A$ and the finite number of locations of~$\A$, there is a finite number of possible symbolic states reachable in~$\A$.
%
%
\end{proof}

Let us compute below an upper bound on the size of this symbolic graph.
In~the following, $|\LT|$ denotes the number of different parametric linear terms (\ie{} the number of guards and invariants) used in~$\A$.

\begin{restatable}{proposition}{PropositionZG}
\label{proposition:complexity-zone-graph}
	The parametric zone graph of a 1cPTA is in $|\Loc| \times  2^{|\LT|(|\LT|+1)}$.
\end{restatable}


\begin{proof}
	First, note that, given a parametric linear term~$\lt_i$, an inequality $\clock \compOp \lt_i$ cannot be conjuncted with other $\clock \compOp' \lt_i$, where $\mathord{\compOp} \neq \mathord{\compOp'}$ (unless $\mathord{\compOp} = \mathord\geq$ and $\mathord{\compOp'} = \mathord\leq$ or the converse, in which case the conjunction is equivalent to a single equality).
	Hence, given $\lt_i$, a 1-clock symbolic constraint contains only one inequality of the form $\clock \compOp \lt_i$.
	The same reasoning applies to parametric inequalities $\lt^1_j \compOp \lt^2_j$.
	
	There are $|\LT|$ different linear terms in~$\A$, and hence $|\LT|$ different inequalities of the form $\clock \compOp \lt_i$ to be used in a 1-clock symbolic constraint.
	Following the same reasoning, there are $|\LT|^2$ different inequalities of the form $\lt^1_j \compOp \lt^2_j$.
	
	Hence the set $\oneClockSC(\A)$ contains $2^{|\LT|} \times 2^{|\LT|^2} = 2^{|\LT|(|\LT|+1)}$ elements.
	
	These constraints can be met for each of the $|\Loc|$ locations.
	This gives that the zone graph of~$\A$ contains at most $|\Loc| \times  2^{|\LT|(|\LT|+1)}$ symbolic states.
\end{proof}

\subsection{Decidability and Synthesis for Deterministic 1-clock PTA}\label{ss:dec-1clock}

We show here that the language- and trace-preservation problems are
decidable for deterministic 1cPTA. These results rely on the correctness and
completeness of \cref{algo:IM} and on the finiteness of the parametric
zone graph of 1cPTA.

\begin{theorem}[trace-preservation synthesis]
\label{theorem:trace-preservation-synthesis}
	Let $\A$ be a deterministic 1cPTA and $\pval$ be a parameter valuation.
	The set of parameters for which the trace set is the same as in $\valuate{\A}{\pval}$ is computable in time proportional to $|\Loc| \times  2^{|\LT|(|\LT|+1)}$.
\end{theorem}
\begin{proof}
	Since $\A$ is a 1cPTA, then its parametric zone graph is finite from \cref{theorem:finiteness-zone-graph}.
	Hence $\TPS(\A, \pval)$ terminates.
	Furthermore, since $\A$ is deterministic, from \cref{theorem:TPS:correctness,theorem:TPS:completeness}, $\TPS(\A, \pval)$ returns all parameter valuations $\pval'$ such that $\Traces(\valuate{\A}{\pval'}) = \Traces(\valuate{\A}{\pval})$.
	
	Concerning the complexity, in the worst case, all symbolic states of~$\A$ are $\pval$-compatible, and $\TPS(\A, \pval)$ needs to explore the entire parametric zone graph, which is of size $|\Loc| \times  2^{|\LT|(|\LT|+1)}$.
\end{proof}

%

\begin{theorem}[language-preservation synthesis]\label{theorem:language-preservation-synthesis}
	Let $\A$ be a deterministic 1cPTA and $\pval$ be a parameter valuation.
	The set of parameters for which the language is the same as in $\valuate{\A}{\pval}$ is computable in $|\Loc| \times  2^{|\LT|(|\LT|+1)}$.
\end{theorem}
\begin{proof}
	Since $\A$ is deterministic, the set of parameter valuations $\pval'$ such that $\Lang(\valuate{\A}{\pval'}) = \Lang(\valuate{\A}{\pval})$ is the same as the set of parameter valuations $\pval'$ such that $\Traces(\valuate{\A}{\pval'}) = \Traces(\valuate{\A}{\pval})$.
	Hence one can directly apply $\TPS(\A, \pval)$ to compute the parameter valuations with the same language as~$\valuate{\A}{\pval}$.
\end{proof}

As direct corollaries of these results, the language- and trace-preservation problems are decidable for deterministic 1cPTA, and so are their \emph{continuous} (robust) counterparts.\ea{28/12/2018:  j'ai ajouté que la version robuste est aussi décidable ; puisqu'il suffit de tester la robustesse sur le polyèdre ; besoin d'explication ?}




\begin{remark}\label{remark:incomplete:1cPTA}
	$\TPS$ is not complete for non-deterministic 1cPTA: in fact, the PTA in \cref{figure:example:nondet} is a 1cPTA, and therefore \cref{remark:incompleteness} applies here too.
\end{remark}

\subsection{Undecidability for L/U-PTA}\label{ss:L/U}

We showed so far that the language- and trace-preservation problems
are undecidable for general PTA (\cref{section:undecidable}) and decidable for
(deterministic) 1-clock PTA (\cref{ss:dec-1clock}). These results match
the EF-emptiness problem, also undecidable for general PTA~\cite{AHV93} and
decidable for 1-clock PTA. 
We now show that the situation is different for L/U-PTA (PTA in~which each
parameter is always either used as a lower bound or always as an upper
bound~\cite{HRSV02}): while EF-emptiness is decidable for
L/U-PTA~\cite{HRSV02,BlT09}, we show  that the language-
and trace-preservation problems are not.


	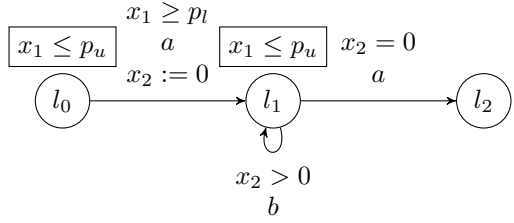
\begin{figure}[tb]
	\centering
	\small
		\begin{tikzpicture}[->, >=stealth', auto, thin,scale=1.4]

			\node[rond] (l0) at(0,0) {$\loc_0$};
			\node[rond] (l1) at (2,0) {$\loc_1$};
			\node[rond] (l2) at (4,0) {$\loc_2$};
			\node[invariant] at (l0.north) {$\clock_1 \leq \param_u$};
			\node[invariant] at (l1.north) {$\clock_1 \leq \param_u$};
			
			\path
				(l0) edge[] node {\begin{tabular}{c}
						$\clock_1 \geq \param_l $\\
						$a$ \\
						$\clock_2 := 0$ \\
					\end{tabular}} (l1)
				(l1) edge[loop below,looseness=6] 
                                node[below] 
                                {\begin{tabular}{c}
						$\clock_2 > 0$ \\
						$b$ \\
					\end{tabular}} (l1)
				(l1) edge[] node[above] {\begin{tabular}{c}
						$\clock_2 = 0$\\
						$a$ \\
					\end{tabular}} (l2)
				;
		\end{tikzpicture}
	\caption{PTA gadget ensuring $\param_l = \param_u$}
	\label{figure:lu:l=u}
	\end{figure}

\paragraph{Constraining parameter equality.}
We~first show how to encode equality of a lower-bound parameter and an
upper-bound parameter in a L/U-PTA, using language preservation.
Consider the PTA gadget depicted in \Cref{figure:lu:l=u}.
Assume a parameter valuation $\pval$ such that $\param_l = \param_u$.
Note that since $\param_l = \param_u$, no time can elapse in $\loc_1$, and the $b$ transition can never be taken.
In fact, we have that the language of this gadget is $aa$ iff $\param_l
= \param_u$.

Now, one can rewrite the two-counter machine encoding of \Cref{ss:language-preservation} using an L/U-PTA which, together with the previous gadget, gives the following undecidability result.

\begin{restatable}{theorem}{TheoremUndecLgLUPTA}
\label{thm-language:L/U-PTA}
The language-preservation problem is undecidable for L/U-PTA with at least one
lower-bound and at least one upper-bound parameter.
\end{restatable}

First, the PTA gadget depicted in \cref{figure:lu:l=u} can be characterized in the following lemma.

\begin{lemma}\label{lemma:lu:l=u}
	In the PTA gadget depicted in \Cref{figure:lu:l=u}, $\loc_2$ is reachable and $b$ can never occur iff $\param_l = \param_u$.
\end{lemma}
\begin{proof}
	\begin{itemize}[align=left]
		\item[($\Rightarrow$)]
			Assume $\loc_2$ is reachable; hence, from the guards and invariants, we necessarily have $\param_l \leq \param_u$.
			Furthermore, $b$ can occur iff it is possible to stay a non-null duration in~$\loc_1$ iff $\param_l < \param_u$.
			Hence, $b$ cannot occur implies $\param_l \geq \param_u$.
			
		\item[($\Leftarrow$)]
			Assume $\param_l = \param_u$.
			Then no time can elapse in $\loc_1$, and hence $b$ cannot occur.
			Furthermore, $\loc_2$ is obviously reachable for any such parameter valuation.
      \qedhere
	\end{itemize}
\end{proof}

We can now prove \cref{thm-language:L/U-PTA}.


\begin{proof}
	The proof is based on the reduction from the halting problem of a two-counter machine.
	The construction encodes the two-counter machine using an L/U-PTA with 2 parameters.
	
	First, let us rewrite the two-counter machine encoding of \Cref{ss:language-preservation} using L/U-PTA as follows.
	We split the parameter $\param$ used in the PTA~$\A$ in the proof of \cref{thm-language} into two parameters $\param_l$ and $\param_u$.
	Any occurrence of $\param$ as an upper-bound (resp.\ lower-bound) in a constraint is replaced with $\param_u$ (resp.\ $\param_l$).
	Equalities of the form $\param = \clock + c$ are replaced with $\param_l \leq \clock + c \land \param_u \geq \clock + c$.
	
	Then, we plug the gadget in \Cref{figure:lu:l=u} before the initial location of our modified encoding of the proof of \cref{thm-language}; more precisely, we fuse $\loc_2$ in \cref{figure:lu:l=u} with $\sinit$ in \cref{figure:language-preservation}, and we reset all clocks in the transition from $\loc_1$ to~$\loc_2$.
	This gives a new PTA, say~$\ALU$.
	
	Let $\pval$ be the reference parameter valuation such that $\param_l = \param_u = 0$.
	For~$\pval$, the language of the gadget of \Cref{figure:lu:l=u} is $aa$.
	Recall that in the proof of \cref{thm-language}, the language of $p=0$ is $a^\omega$, and hence the language of our modified PTA $\ALU$ for $\pval$ is $aaa^\omega = a^\omega$.
	
	Suppose the two-counter machine does not halt, and consider a parameter valuation $\pval' \neq \pval$.
	If $\param_l \neq \param_u$ in $\pval'$, then from \cref{lemma:lu:l=u}, the language of the gadget for~$\pval'$ is either a single deadlocked $a$ (if $\param_l > \param_u$), or $ab^\omega | aa$ (if $\param_l < \param_u$); in both cases, the language of~$\valuate{\ALU}{\pval'}$ differs from the language of~$\valuate{\ALU}{\pval}$ (that is $a^\omega$).
	If $\param_l = \param_u$, then we fall in the situation of \cref{thm-language}: that is, there is no way for $\pval'$ to accept the same language as~$\pval$.
	Hence there exists no parameter valuation $\pval' \neq \pval$ such that the language is the same as for~$\pval$.
	
	Conversely, suppose the two-counter machine halts, and consider a parameter valuation $\pval' \neq \pval$.
	Again, if $\param_l \neq \param_u$ in $\pval'$,
	then the language necessarily differs from~$\pval$.
	If $\param_l = \param_u$, then we fall again in the situation of \Cref{ss:language-preservation}: for some $\pval' \neq \pval$ such that $\param_l = \param_u$ and $\param_l$ is large enough to encode the two counters maximum value, then the language is the same as for~$\pval$.
	Hence there exists a parameter valuation $\pval' \neq \pval$ such that the language is the same as for~$\pval$.
	
	As a consequence, the two-counter machine halts iff there exists a parameter valuation $\pval' \neq \pval$ such that the language is the same as for~$\pval$.     
\end{proof}

This reasoning can be reused to prove the undecidability for L/U-PTA of the
other problems considered in \cref{section:undecidable}. 
It~follows:
\begin{theorem}\label{theorem:all-undecidable:LUPTA}
	\begin{enumerate}
		\item The trace-preservation problem is undecidable for L/U-PTA with at least one lower-bound and at least one upper-bound parameter.
		\item The robust language- and trace-preservation problems are undecidable for L/U-PTA with at least one lower-bound and at least one upper-bound parameter.
	\end{enumerate}
\end{theorem}

\subsection{A Decidability Result for 1-Parameter L-PTA and U-PTA}\label{ss:L-U-PTA}






\looseness=-1
In~\cite{BlT09}, a bound is exhibited for both L-PTA and U-PTA (\ie{} PTA with only lower-bound, resp.\ upper-bound, parameters) such that
either all parameter valuations beyond this threshold have an accepting run, 
or none of them has. 
This provides an algorithm for synthesizing
all integer parameter valuations for which there exists an accepting run,
by considering this bound, and then enumerate all (integer) valuations below this
bound.


Unfortunately, such a bound for U-PTA (and L-PTA) does not exist for the language.
Consider the U-PTA in \cref{fig:UPTA:lg-diff}. Then,  given ${\param \in \grandn}$, the accepted language is $a^{\leq \param} b^\omega$.
Hence, it differs for all integer values of~$\param$.
For L-PTA, the situation is similar: the language of the L-PTA in \cref{fig:LPTA:lg-diff} is $a^{\param}a^* b^\omega \cup a^\omega$, \ie{} at least $\param$ times $a$ followed (if the number of $a$ is finite) by an infinite number of~$b$.

\begin{figure}
	\begin{subfigure}[b]{0.45\textwidth}
	\centering
		\begin{tikzpicture}[->, >=stealth', auto, node distance=3.5cm, thin]
			\node[rond] (l1) {$\loc_1$};
                        \draw[latex'-] (l1.-135) -- +(-135:3mm);
			\node[rond, right of=l1] (l2) {$\loc_2$};
			
			\path
				(l1) edge[loop
                          above,looseness=6,out=120,in=60] node
                            {\begin{tabular}{c}$ \clock_1 = 1 $; $a$\\[-1pt]$ \clock_1 := 0$\end{tabular}} (l1)
				(l1) edge[] node {\begin{tabular}{c}$ \clock_2
                                \geq \param \land \clock_1 = 0$\\[-1pt] $b$\end{tabular}} (l2)
				(l2) edge[loop above,looseness=6,out=120,in=60] node {$b$} (l2)
			;
		\end{tikzpicture}
		\caption{L-PTA}
		\label{fig:LPTA:lg-diff}
		
	\end{subfigure}%
	\hfill
	\begin{subfigure}[b]{0.45\textwidth}
	\centering
		\begin{tikzpicture}[->, >=stealth', auto, node distance=2.5cm, thin]
			\node[rond] (l1) {$\loc_1$};
                        \draw[latex'-] (l1.-135) -- +(-135:3mm);
			\node[rond, right of=l1] (l2) {$\loc_2$};
			
			\path
				(l1) edge[loop
                          above,looseness=6,out=120,in=60] node
                            {\begin{tabular}{c}$ \clock_1 = 1 \land \clock_2
                                \leq \param$\\[-1pt]$a$; $ \clock_1 := 0$\end{tabular}} (l1)
				(l1) edge[] node {\begin{tabular}{c}$ \clock_2 \leq \param$\\[-1pt]$b$\end{tabular}} (l2)
				(l2) edge[loop above,looseness=6,out=120,in=60] node {$b$} (l2)
			;
		\end{tikzpicture}
		
		\caption{U-PTA}
		\label{fig:UPTA:lg-diff}
	\end{subfigure}

	\caption{An L-PTA and a U-PTA for which the language differs for all $\param \in \grandn$}
	\label{fig:LUPTA:lg-diff}
\end{figure}
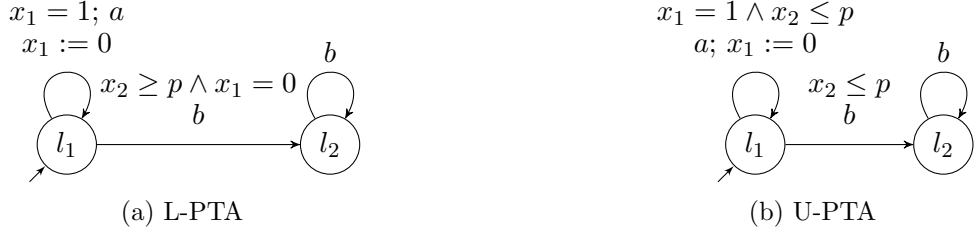

We now show that the trace-preservation problem is decidable for
deterministic L-PTA and U-PTA with a single integer parameter and arbitrarily
many clocks: given a reference integer parameter valuation~$\pval$, it
suffices to check $\pval+1$ and $\pval-1$ to decide whether another parameter valuation yields the same trace set as~$\pval$.

\newcommand{\enonceTheoremDecLUoneip}{
   The trace-preservation problem is \PSPACE-complete for deterministic U-PTA and deterministic L-PTA with a single integer-valued parameter.
}
\begin{theorem}\label{theorem:trace-preservation:l/u}
	\enonceTheoremDecLUoneip{}
\end{theorem}
\begin{proof}
	Let $\A$ be a deterministic U-PTA with a single integer-valued parameter~$\param$ (the reasoning is dual for L-PTA).
	Let $\pval$ be a valuation of~$\param$.
	Construct the trace set of $\valuate{\A}{\pval}$.
	Consider the valuation $\pval+1$ (\ie{} the smallest integer valuation larger than $\pval$).
	It is known that increasing a parameter in a U-PTA can only \emph{add} behaviors.
	Suppose $\valuate{\A}{(\pval+1)}$ adds a behavior, \ie{} enables a transition that was not enabled in $\valuate{\A}{\pval}$.
    Since $\A$ is deterministic, then necessarily $\valuate{\A}{(\pval+1)}$ contains a transition $\loc_1 \Fleche{\action} \loc_2$ that did not exist in $\valuate{\A}{\pval}$.\ea{point perfectible; peut-être vaudrait-il mieux raisonner formellement par induction sur la longueur  des runs ?}\dl{moui, c'est compréhensible mais pas très joli effectivement}
	Hence the trace set of $\valuate{\A}{(\pval+1)}$ strictly contains the trace set of $\valuate{\A}{\pval}$, and the trace set of any valuation greater or equal to $\pval+1$ will again strictly contain the trace set of $\valuate{\A}{\pval}$.
	Hence, deciding whether there exists a valuation greater than $\pval$ for which the trace set is the same as $\valuate{\A}{\pval}$ is equivalent to checking whether the trace set of $\valuate{\A}{(\pval+1)}$ is the same as the trace set of $\valuate{\A}{\pval}$.
	
The proof for $\pval-1$ is symmetric.	
	Hence it is decidable whether there exists a valuation different from~$\pval$ for which the trace set is the same as $\valuate{\A}{\pval}$.

    Now, for the \PSPACE upper bound, we actually prove that testing the inclusion of the untimed language of timed automaton $\A_1$ in the untimed language of a \emph{deterministic} timed automaton $\A_2$ can be done in \PSPACE. Trace set inclusion can then be checked with untimed language inclusion.
    The proof is very similar to that of~\cite{AD94} for \emph{timed language inclusion}.
    
    \begin{lemma}\label{lemma:untimed-language-preservation-det-TA}
		The untimed language inclusion problem for deterministic timed automata is in \PSPACE.
    \end{lemma}
    \begin{proof}
    We build a non-deterministic Turing machine that guesses a path in the product of the two automata. We store on the tape the current state, \ie{} current location and region, as well as the next state and the action leading to it when they are non-deterministically guessed. We also need a counter for the maximum number of steps allowed for the path. When a new state and action are guessed, the machine verifies that the transition is indeed possible in both automata. If it is not possible in $\A_1$ then the machine rejects. If it is possible for $\A_1$ but not $\A_2$, or if the location is accepting in $\A_1$ but not in $\A_2$, we have found a witness for non-inclusion and the machine accepts. If it is possible for both automata, the machine overwrites the current state with the new state, increments the counter and proceeds to guessing a new successor, unless the counter has reached its maximum value, which is the product of the number of states of the region automata of $\A_1$ and $\A_2$. In this last case, the machine also rejects. Since, if an untimed word is accepted by $\A_1$ and not by $\A_2$, there must also be one such word with length less than the maximal value of the counter, it is clear that the machine accepts if and only if the untimed language of $\A_1$ is not included in that of $\A_2$. Finally, storing both states and actions can be done in polynomial space. As for the value of the counter, since its maximum value is exponential in the size of the problem, we need only a polynomial number of bits to store it in binary. So the procedure works in \NPSPACE, and by Savitch's Theorem~\cite{Sav70}, in \PSPACE. \dl{on pourrait détailler un peu moins}\ea{non, moi ça me va comme ça !}
    \end{proof}

    Finally, \PSPACE hardness is obtained by remarking that we can reduce
    reachability in timed automata to the trace preservation problem:
    Consider a deterministic timed automaton $\A$ (without parameter) and one of its location $\ell$.
    Add a parameter
    $p$, a fresh clock $x$, and a self-loop on $\ell$, with guard $p\leq x\leq
    0$. Finally for every transition arriving in $\ell$, add a reset of $x$.
    The added transition is therefore possible only for $p=0$.  Then, it is
    clear that (i) the resulting PTA is an L-PTA with a single parameter that
    we can consider to be integer-valued, and (ii) there exists a value for $p$
    different from $0$ with the same trace set as for $p=0$ if and only if
    $\ell$ is not reachable in $\A$. 
    
    To get a U-PTA instead of an L-PTA, we
    can use guard $1\leq x \leq p$. Again, we use the reference value $p=0$ and
    the transition is this time not possible for $p=0$ but it is for all other integer
    parameter values, so the result follows.
\end{proof}


Since we have a direct correspondence between trace sets and languages in
deterministic automata, we~get:
\begin{theorem}
\label{theorem:language-preservation:l/u}
   The language-preservation problem is \PSPACE-complete for deterministic U-PTA and deterministic L-PTA with a single integer-valued parameter.
\end{theorem}


\cref{theorem:language-preservation:l/u} cannot be lifted to the language for
non-deterministic L- and U-PTA. Consider the U-PTA in
\cref{fig:UPTA:counterex:language}: 
for $\param = 1$, the language is $a
b^\omega$. For $\param = 2$, the language is $a b^\omega | a$, which is
different from $\param = 1$. But then for $\param \geq 3$, the language is
again $a b^\omega$. Hence testing only $\pval + 1 = 2$ is not enough, and the
decidability in this case remains open.

Also note that these non-deterministic PTAs are \emph{deterministic} with the usual definition of determinism (that two outgoing guards can have the same action label provided the guards are mutually disjoint---which is the case here with $x=1$ and $x=2$), which justifies our stronger definition of non-determinism.\label{newtext:nondeterminism}

\begin{figure}
	\begin{subfigure}{0.45\textwidth}
	\centering
		\begin{tikzpicture}[->, >=stealth', auto, node distance=2.5cm, thin]
			\node[rond, initial] (l1) {$\loc_1$};
			\node[rond] (l2) at +(30:2.3cm) {$\loc_2$};
			\node[rond] (l3) at +(-30:2.3cm) {$\loc_3$};
			\node[rond, right of=l3,node distance=2.8cm] (l4) {$\loc_4$};
			
			\path
				(l1) edge[bend left] node[above left=-3mm] {\begin{tabular}{c}$ \clock = 1 \land \clock \leq \param $\\$a$\end{tabular}} (l2)
				(l2) edge[loop right,looseness=6,out=-30,in=30] node {$b$} (l2)
				(l1) edge[bend right] node[below left] {\begin{tabular}{c}$ \clock = 2 \land \clock \leq \param $\\$a$\end{tabular}} (l3)
				(l3) edge[] node[below=2mm] {\begin{tabular}{c}$ \clock = 3 \land \clock \leq \param $\\$b$\end{tabular}} (l4)
				(l4) edge[loop right,looseness=6,out=-30,in=30] node {$b$} (l4)
			;
		\end{tikzpicture}
		
		\caption{Language}
		\label{fig:UPTA:counterex:language}
	\end{subfigure}
	\hfill
	\begin{subfigure}{0.45\textwidth}
	\centering
		\begin{tikzpicture}[->, >=stealth', auto, node distance=2.5cm, thin]
			\node[rond, initial] (l1) {$\loc_1$};
			\node[rond, right of=l1] (l2) {$\loc_2$};
			\node[rond, right of=l2] (l3) {$\loc_3$};
			
			\path
				(l1) edge[bend right, below] node {\begin{tabular}{c}$ \clock = 2 \land \clock \leq \param $\\[-1mm]$a$\end{tabular}} (l2)
				(l1) edge[bend left, above] node {\begin{tabular}{c}$ \clock = 1 \land \clock \leq \param $\\$a$\end{tabular}} (l2)
				(l2) edge[bend right, below] node {\begin{tabular}{c}$ \clock = 3 \land \clock \leq \param  $\\$b$\end{tabular}} (l3)
				(l2) edge[bend left, above] node {\begin{tabular}{c}$ \clock = 1 $\\$b$\end{tabular}} (l3)
				(l3) edge[loop above,,looseness=6,out=120,in=60] node {$b$} (l3)
			;
                        \path[use as bounding box] (0,-2.3);
		\end{tikzpicture}
		
		\caption{Traces}
		\label{fig:UPTA:counterex:traces}
	\end{subfigure}

	\caption{Counterexamples for the method to decide the trace preservation emptiness in U-PTA}
	\label{fig:UPTA:counterex}
\end{figure}
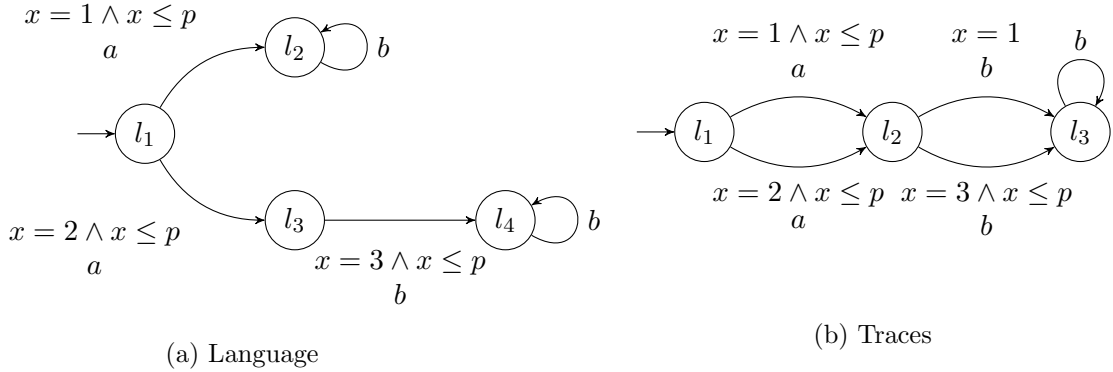

Similarly, \cref{theorem:trace-preservation:l/u} cannot be lifted to the trace preservation in non-deterministic PTA, as witness in the U-PTA in \cref{fig:UPTA:counterex:traces}.

%


\section{Conclusion and Perspectives}\label{section:conclusion}

In this paper, we studied the decidability of the language- and
trace-preservation problems in parametric timed automata. We summarize
in \cref{table:summary} our (un)\nobreak\hskip0pt\relax decidability results for PTA and its
subclasses with arbitrarily many clocks; a red italicized cell denotes undecidability while a green plain cell denotes decidability.
(1ip-dL\&U-PTA stand for deterministic L-PTA, resp.\ U-PTA, with one integer-valued parameter;
	L\&U-PTA stand for L-PTA and U-PTA with rational-parameters or more than one clock or without our determinism assumption;
	bPTA stand for PTA with bounded parameters;
	d1cPTA stands for deterministic 1-clock PTA.)
We also showed
that both problems are decidable for deterministic PTA with a single
clock.  ``N/A'' indicates a problem that is not relevant for this
class (robust versions of our problems are not so relevant for
integer-valued parameters).



\newcommand\crefabbr[1]{%
\begingroup
	\crefname{theorem}{\text{Th.}}{\text{Th.}}
	\crefname{corollary}{\text{Cor.}}{\text{Cor.}}
	\cref{#1}
\endgroup%
}

\newcommand{\colCellDec}{\cellcolor{green!20}}
\newcommand{\colCellDecNous}{\cellcolor{green}}
\newcommand{\colCellUndec}{\cellcolor{red!20}\em}
\newcommand{\colCellUndecNous}{\cellcolor{red}\em}
\newcommand{\colCellOpen}{\cellcolor{yellow!20}open}
\newcommand{\colCellNA}{\cellcolor{black!20}N/A}
\newcommand{\cellHeader}[1]{\cellcolor{blue!40}\textbf{#1}}

\begin{table}[b]
	\centering
	\scriptsize

	\setlength{\tabcolsep}{.2em}
	
	\begin{tabular}{| c | c | c | c | c | c | c | c | c |}
		\hline
		\cellHeader{Preservation} & \cellHeader{1ip-dL\&U-PTA} & \cellHeader{L\&U-PTA} & \cellHeader{bL/U-PTA} & \cellHeader{L/U-PTA} & \cellHeader{d1cPTA} & \cellHeader{bPTA} & \cellHeader{PTA} \\
		
		\hline
		
		\cellHeader{Language} & \colCellDecNous{}\crefabbr{theorem:language-preservation:l/u} & \colCellOpen{} & \colCellUndecNous{}\crefabbr{theorem:all-undecidable:LUPTA} & \colCellUndecNous{}\crefabbr{thm-language:L/U-PTA} & \colCellDecNous{}\crefabbr{theorem:language-preservation-synthesis} & \colCellUndecNous{}\crefabbr{theorem:robust-language-preservation:undecidable:bPTA} & \colCellUndecNous{}\crefabbr{thm-language} \\
		
		\hline
		
		\cellHeader{Trace} & \colCellDecNous{}\crefabbr{theorem:trace-preservation:l/u} & \colCellOpen{} & \colCellUndecNous{}\crefabbr{theorem:all-undecidable:LUPTA} & \colCellUndecNous{}\crefabbr{theorem:all-undecidable:LUPTA} & \colCellDecNous{}\crefabbr{theorem:trace-preservation-synthesis} & \colCellUndecNous{}\crefabbr{theorem:robust-trace-preservation:undecidable:bPTA} & \colCellUndecNous{}\crefabbr{theorem:trace-preservation-undec} \\
		
		\hline
		
		\cellHeader{Robust language} & \colCellNA{} & \colCellOpen{} & \colCellUndecNous{}\crefabbr{theorem:all-undecidable:LUPTA} & \colCellUndecNous{}\crefabbr{theorem:all-undecidable:LUPTA} & \colCellDecNous{}\crefabbr{theorem:language-preservation-synthesis} & \colCellUndecNous{}\crefabbr{theorem:robust-language-preservation:undecidable:bPTA} & \colCellUndecNous{}\crefabbr{theorem:robust-language-preservation:undecidable:bPTA} \\
		
		\hline
		
		\cellHeader{Robust trace} & \colCellNA{} & \colCellOpen{} & \colCellUndecNous{}\crefabbr{theorem:all-undecidable:LUPTA} & \colCellUndecNous{}\crefabbr{theorem:all-undecidable:LUPTA} & \colCellDecNous{}\crefabbr{theorem:trace-preservation-synthesis} &  \colCellUndecNous{}\crefabbr{theorem:robust-trace-preservation:undecidable:bPTA} & \colCellUndecNous{}\crefabbr{theorem:robust-trace-preservation:undecidable:bPTA}\\
		\hline
	\end{tabular}

	\caption{Decidability of preservation emptiness problems for subclasses of PTA}
    \label{table:summary}
\end{table}

\paragraph{Future Works.}
First, we used an \emph{ad-hoc} encoding of a 2-counter machine for
our  undecidability proofs, using four parametric clocks.
In~contrast, a new encoding of a 2-counter machine using PTA was proposed recently in \cite{BBLS15} to show the undecidability of the EF-emptiness problem for integer-valued parameters, and that makes use of only three parametric clocks.
It remains open whether the (non-robust) problems considered in our manuscript could be proved undecidable with as few as three parametric clocks in the case of integer-valued parameters.
In addition, it remains to be proved whether the trace preservation problem is undecidable for a bounded number of clocks and without diagonal constraints.\ea{rajouté ça (24/1/19)}

Concerning the decidability for a single clock, it would be
interesting to study whether this result can be adapted to PTA with
one parametric clock and arbitrarily many non-parametric clocks,
following the corner-point abstraction recently used in the
construction of~\cite{BBLS15} to show the decidability of the
EF-emptiness problem.

\smallskip

Language-preservation problems have been considered
in~\cite{Sankur-MFCS11,AHJR12} in the setting of guard enlargement
(for timed automata and time Petri nets): this is a weaker setting, in
which the single parameter~$\epsilon$ can only be used under the forms
$x\geq a-\epsilon$ and $x\leq b+\epsilon$. This makes the robust
version of the language-preservation problem decidable. In a similar
flavor, time-abstract simulation of shrunk timed
automata~\cite{fsttcs2011-SBM} also shares commonalities with the
problem we studied in the present paper. Identifying larger classes of PTA
with decidable language-preservation problems in the light of these
results is a relevant direction for future research.

Finally, we showed in~\cite{AL18} that some of the results presented in this paper extend to the smaller class of parametric event-recording automata~\cite{ALin17}, \ie{}, the language preservation problem remains undecidable in that setting.
It remains however to prove whether the trace preservation problem is decidable for this subclass.

\section*{Acknowledgment}
\noindent
We are grateful to 
Olivier~H.~Roux 
and to the anonymous reviewers for useful comments.


	\newcommand{\CCIS}{Communications in Computer and Information Science}
	\newcommand{\ENTCS}{Electronic Notes in Theoretical Computer Science}
	\newcommand{\FAC}{Formal Aspects of Computing}
	\newcommand{\FI}{Fundamenta Informaticae}
	\newcommand{\FMSD}{Formal Methods in System Design}
	\newcommand{\IJFCS}{International Journal of Foundations of Computer Science}
	\newcommand{\IJSSE}{International Journal of Secure Software Engineering}
	\newcommand{\IPL}{Information Processing Letters}
	\newcommand{\JAIR}{Journal of Artificial Intelligence Research}
	\newcommand{\JLAP}{Journal of Logic and Algebraic Programming}
	\newcommand{\JLAMP}{Journal of Logical and Algebraic Methods in Programming} 
	\newcommand{\JLC}{Journal of Logic and Computation}
	\newcommand{\LMCS}{Logical Methods in Computer Science}
	\newcommand{\LNCS}{Lecture Notes in Computer Science}
	\newcommand{\RESS}{Reliability Engineering \& System Safety}
	\newcommand{\STTT}{International Journal on Software Tools for Technology Transfer}
	\newcommand{\TCS}{Theoretical Computer Science}
	\newcommand{\ToPNoC}{Transactions on Petri Nets and Other Models of Concurrency}
	\newcommand{\TSE}{{IEEE} Transactions on Software Engineering}

\bibliographystyle{alpha}
\bibliography{bibexport}


\ifcomments
  \addtocontents{toc}{\endgroup}
\fi

\ifdefined\WithReply
	\input{letter2.tex}
\fi

\end{document}